\documentclass[11pt, reqno]{amsart}
\usepackage[utf8]{inputenc}
\usepackage[T1]{fontenc}
\usepackage{amssymb}
\usepackage{graphics}
\usepackage{xcolor}
\usepackage{enumitem}
\usepackage[pagebackref, colorlinks = true, linkcolor = teal, urlcolor  = teal, citecolor = violet]{hyperref}
\usepackage[margin=1in]{geometry}
\usepackage{cleveref}
\crefname{equation}{Eq.}{Eqs.}
\usepackage{bbm}
\usepackage{pifont}

\usepackage[normalem]{ulem}

\numberwithin{equation}{section}

\newtheorem{theorem}{Theorem}[section]
\newtheorem{definition}[theorem]{Definition}

\newtheorem{proposition}[theorem]{Proposition}
\newtheorem{corollary}[theorem]{Corollary}
\newtheorem{lemma}[theorem]{Lemma}
\newtheorem{remark}[theorem]{Remark}

\newcommand{\id}{{\rm Id}}

\newcommand{\one}{{\mathbf 1}}
\newcommand{\tr}{{\operatorname{tr}}}

\renewcommand{\d}{{\rm d}}

\newcommand{\pp}{{\mathbb P}}

\newcommand{\rr}{{\mathbb R}}
\newcommand{\nn}{{\mathbb N}}
\newcommand{\cc}{{\mathbb C}}
\renewcommand{\P}{{\mathrm P}}

\newcommand{\inv}{_{\mathrm{inv}}}

\newcommand{\as}{\operatorname{-}\mathrm{a.s.}}

 \def\cB{{\mathcal B}} 
  \def\cF{{\mathcal F}}
 \def\cH{{\mathcal H}} 
\def\cJ{{\mathcal J}}  \def\cL{{\mathcal L}}
  \def\cO{{\mathcal O}}
  \def\cR{{\mathcal R}}
\def\cS{{\mathcal S}}

\def\R{{\mathbb R}}
\def\N{{\mathbb N}}

\def\C{{\mathbb C}}
\def\Z{{\mathbb Z}}
\def\E{{\mathbb E}}
\def\P{{\mathbb P}}

\def\at{{\rm at}}
\def\tr{{\rm Tr}}
\def\e{{\rm e}}
\def\d{{\rm d}}
\def\at{{\rm at}}
\def\inv{{\rm inv}}
\def\Ran{{\rm Ran}}

\def\ss{{\sigma\sigma'}}

\newcommand{\bra}{\langle} 
\newcommand{\ket}{\rangle}
\newcommand{\eps}{\varepsilon}

\newcommand{\ds}{\displaystyle}
\def\one{{\mathchoice {\rm 1\mskip-4mu l} {\rm 1\mskip-4mu l} {\rm 1\mskip-4.5mu l} {\rm 1\mskip-5mu l}}}

\makeatletter
\newcommand{\specialcell}[1]{\ifmeasuring@#1\else\omit$\displaystyle#1$\ignorespaces\fi}
\makeatother

\begin{document}

\title{Quantum Trajectory of the One Atom Maser}

\author{T. Benoist}
\address{Institut de Math\'ematiques de Toulouse, \'Equipe de Statistique et Probabilit\'es,
Universit\'e Paul Sabatier, 31062 Toulouse Cedex 9, France}
\email{tristan.benoist@math.univ-toulouse.fr}
\author{L. Bruneau}
\address{D\'epartement de Math\'ematiques, CNRS - UMR 8088,
CY Cergy Paris Universit\'e, 95000 Cergy-Pontoise, France.}
 \email{laurent.bruneau@u-cergy.fr}
\author{C. Pellegrini}
\address{Institut de Math\'ematiques de Toulouse, \'Equipe de Statistique et Probabilit\'es,
Universit\'e Paul Sabatier, 31062 Toulouse Cedex 9, France}
\email{clement.pellegrini@math.univ-toulouse.fr}

\subjclass[2000]{}

\begin{abstract} The evolution of a quantum system undergoing repeated indirect measurements naturally leads to a Markov chain on the set of states which is called a quantum trajectory. In this paper we consider a specific model of such a quantum trajectory associated to the one-atom maser model. It describes the evolution of one mode of the quantized electromagnetic field in a cavity interacting with two-level atoms. When the system is non-resonant we prove that this Markov chain admits a unique invariant probability measure. We moreover prove convergence in the Wasserstein metric towards this invariant measure. These results rely on a purification theorem: almost surely the state of the system approaches the set of pure states. Compared to similar results in the literature, the system considered here is infinite dimensional. While existence of an invariant measure is a consequence of the compactness of the set of states in finite dimension, in infinite dimension existence of an invariant measure is not free. Furthermore usual purification criterions in finite dimension have no straightforward equivalent in infinite dimension.\end{abstract}

\date{\today}

\maketitle



\tableofcontents

\section{Introduction}

Quantum trajectories describe the evolution of a quantum system undergoing repeated and indirect measurements. An indirect measurement means that a direct measurement is performed on an auxiliary system, called probe or ancilla, after it has interacted for some time with the system we are interested in, see e.g. \cite{HR06}. The physical paradigm of quantum trajectories is probably the one atom maser model \cite{FJM86a, CDG92, MWM85, WVHW99, WBKM00}, and some of its subsequent elaborations  \cite{DRBH87, GK07, RH05, RBH01}. Here, the system of interest is the quantized electromagnetic field in a cavity through which a beam of atoms, the probes, is shot in such a way that, at least with very high probability \cite{HBR97}, no more than {\sl one atom} is present in the cavity at any time. Such systems play a fundamental role in the experimental and theoretical investigations of basic matter-radiation processes. They are also of practical importance in quantum optics and quantum state engineering \cite{MWM85,WVHW99,WBKM00,RH05,VAS93}. So-called ``One-Atom Masers'', where the beam is tuned in such a way that at each given moment a single atom is inside a microwave cavity and the interaction time $\tau$ is the same for each atom, have been experimentally realized in laboratories \cite{MWM85,WVHW99,HR06}. These systems have also been considered as an effective ressource of non equilibrium free energy or as a heat, work or information reservoir \cite{SSB17}. They are sometimes called Quantum collision models. For a recent review from a more theoretical and experimental physics point of view we refer the reader e.g. to \cite{CLGP22} and references therein.

The one-step ``interaction + indirect measurement'' time evolution of the quantum system is then conditioned to the measurement outcome. More precisely if the possible outcomes of the measurement are labeled by $\omega$ in some alphabet $\cR$, and the system is in the state described by the density operator $\rho$, then $\omega$ is observed with probability $\tr(V_\omega\rho V_\omega^*)$ and the state of the system becomes 
\[
\frac{V_\omega \rho V_\omega^*}{\tr(V_\omega\rho V_\omega^*)},
\]
where the operators $V_\omega$ satisfy the stochasticity condition $\ds \sum_{\omega\in\cR} V_\omega^*V_\omega=\one$.

Iterating this procedure defines a Markov chain $(\rho_t)_{t\in\N}$ on the set of states which is called a quantum trajectory: after $t$ measurements with outcomes $\omega_1,\ldots,\omega_t$  the state of the system becomes
\[
\rho_t = \frac{V_{\omega_t}\cdots V_{\omega_1}\rho_0 V_{\omega_1}^*\cdots V_{\omega_t}^*}{\tr\left(V_{\omega_t}\cdots V_{\omega_1}\rho_0 V_{\omega_1}^*\cdots V_{\omega_t}^*\right)},
\]
and this happens with probability $\tr\left(V_{\omega_t}\cdots V_{\omega_1}\rho_0 V_{\omega_1}^*\cdots V_{\omega_t}^*\right)$. 

On the mathematical side, and to the best of our knowledge, the study of the large time behavior of quantum trajectories goes back to the pioneering works \cite{KM03,KM04,MK06} (for continuous time models one can also consult \cite{Bar09} for a general introduction and \cite{Bar03} for further results). When the quantum system under consideration is finite dimensional, so that the $\rho_t$ are simply non-negative matrices with unit trace, a key result in the theory is a purification theorem obtained by K\"ummerer and Maassen in \cite{MK06}. Provided a suitable purification condition {\bf (Pur)} holds (see Remark~\ref{rem:purificationcondition}, p.\pageref{rem:purificationcondition}) the state $\rho_t$ almost surely approaches the set of pure states (the rank $1$ orthogonal projections which are also the extreme points of the convex set of states). This result has been proven using a slightly different argument in \cite{BFPP19} and then used as a key step to prove the uniqueness of an invariant probability measure for quantum trajectories together with convergence towards this invariant measure. The convergence is proven to be geometric in the Wasserstein metric. A similar result for continuous time models has been proven in \cite{BFPP21}.

In this paper we consider a specific model of a quantum trajectory in which the system under consideration is infinite dimensional. This model describes the one-atom maser experiment mentioned above, see Section~\ref{ssec:masermodel} for a precise description. It was considered in \cite{BP09,Bru14} but only at the level of the averaged state, and without mention to quantum trajectories. The averaged state at time $t$ is given by 
\[
\E(\rho_t) = \sum_{\omega_1,\ldots,\omega_t}  \frac{V_{\omega_t}\cdots V_{\omega_1}\rho_0 V_{\omega_1}^*\cdots V_{\omega_t}^*}{\tr\left(V_{\omega_t}\cdots V_{\omega_1}\rho_0 V_{\omega_1}^*\cdots V_{\omega_t}^*\right)} \times \tr\left(V_{\omega_t}\cdots V_{\omega_1}\rho_0 V_{\omega_1}^*\cdots V_{\omega_t}^*\right)
\]
and can be written as $\ds \E(\rho_t) = \cL^t(\rho_0)$ where $\cL$ denotes the, completely positive and trace preserving, map defined on the set of states by
\[
\cL(\rho) = \sum_{\omega\in\cR } V_\omega \, \rho\, V_\omega^*.
\]
The analysis of the map $\cL$, and in particular its large $t$ limit, is the main subject of \cite{BP09,Bru14}, see Theorem~\ref{thm:Lmixing} at the end of Section~\ref{ssec:masermodel}. In the present article we study similar questions taking into account the measurements outcomes, namely for the associated quantum trajectories.

One of the key notion for the investigation of this model is that of Rabi resonance, see Equation~\eqref{def:rabi}. Generically the system has no such resonance. In this case we prove a purification theorem similar to the one in \cite{MK06} and existence and uniqueness of an invariant probability measure for the quantum trajectory as in the finite dimensional case of \cite{BFPP19}. We moreover prove convergence in the Wasserstein metric towards this invariant measure. The convergence relies in part on the convergence result of the average state $\E(\rho_t)=\cL^t(\rho_0)$. It was proven in \cite{Bru14} that the latter can be arbitrarily slow (this is related to the absence of a spectral gap for the map $\cL$). As a consequence, here too one cannot expect any rate of convergence contrary to what happens in finite dimensional systems. Note also that in finite dimension the set of density matrices is a compact set. One of the crucial consequences is that there always exists at least one invariant measure (the sequence of laws of the Markov chain is indeed tight). In infinite dimension this property is not ensured and we bypass this problem by linking quantum trajectories and the invariant state of $\mathcal L$ with a classical birth and death process. Concerning the purification property we show by hand that quantum trajectories converge towards some subset of the pure states, namely the set of pure Fock states. One of the main ingredients is a convergence theorem for martingales in Banach spaces.
 
When the system possesses Rabi resonances the first immediate consequence is the non uniqueness of invariant probability measures. However this excludes neither purification nor convergence in law towards an invariant measure. We prove partial results in this direction. In particular, when the system possesses at least two pairs of consecutive Rabi resonances (called degenerate case in \cite{BP09}) we prove that purification does not hold.

Finally let us stress that we are only concerned here with the state of the quantum system, and not with the process describing the sequence of measurement outcomes themselves (see however Theorem \ref{thm:mixing-dynsys}). These repeated measurements exhibit a rich mathematical structure. Numerous considerations have been dedicated to limit theorems for the measurement outcomes.
Notably, Law of Large Numbers \cite{KM03}, Central Limit Theorems \cite{AGPS15, vHG15, CP15}
and Large Deviation Principles \cite{vHG15, CP15} have been derived. We refer the reader to e.g. \cite{BJPP18,BCJP21} and references therein for recent developments in this direction. Note that these results deal with finite dimension. We mention \cite{GvHCG} where a large deviation result is proven for the counting process associated to a similar one-atom maser model in infinite dimension.

Our paper is organized as follows. In Section~\ref{sec:model} we describe the one-atom maser model we consider and state our main results, Theorem~\ref{thm:invariantmeasure} concerning the invariant measure and the purification Theorem~\ref{thm:purification}. The proof of Theorem~\ref{thm:purification} is given in Section~\ref{sec:purification}, while Theorem~\ref{thm:invariantmeasure} is proven in Section~\ref{sec:invariantmeasure}. Finally Section~\ref{sec:resonant} is devoted to the resonant situation.

\medskip

\noindent {\bf Acknowledgements.} T. B. and C. P. were supported by the ANR project ``ESQuisses'', grant number ANR-20-CE47-0014-01, the ANR project ``Quantum Trajectories'' grant number ANR-20-CE40-0024-01 and the program ``Investissements d'Avenir'' ANR-11-LABX-0040 of the French National Research Agency. C. P. is also supported by the ANR projects Q-COAST ANR-19-CE48-0003.

%%%%%%%%%%%%%%%%%%%%%%%%%%%%%%%%%%%%%%%%%%%%%%%%%%%%%%%%%%%%%%%%%%%%
%%%%%%%%%%%%%%%%%%%%%%%%%%%%%%%%%%%%%%%%%%%%%%%%%%%%%%%%%%%%%%%%%%%%
%%%%%%%%%%%%%%%%%%%%%%%%%%%%%%%%%%%%%%%%%%%%%%%%%%%%%%%%%%%%%%%%%%%%
%%%%%%%%%%%%%%%%%%%%%%%%%%%%%%%%%%%%%%%%%%%%%%%%%%%%%%%%%%%%%%%%%%%%

\section{Description of the model and main results}\label{sec:model}

%%%%%%%%%%%%%%%%%%%%%%%%%%%%%%%%%%%%%%%%%%%%%%%%%%%%%%%%%%%%%%%%%%%%
%%%%%%%%%%%%%%%%%%%%%%%%%%%%%%%%%%%%%%%%%%%%%%%%%%%%%%%%%%%%%%%%%%%%

\subsection{The one-atom maser model}\label{ssec:masermodel}

The Hilbert space for the cavity field is $$\cH_\cS\equiv\Gamma_+(\C),$$ the Bosonic Fock space over $\C$. Its states are element of the set $\cJ_1$ of density operators on $\cH_\cS$, i.e. the set of positive semidefinite trace class operators on $\cH_\cS$ with unit trace. We equip the space $\cJ(\cH_S)$ of trace class operators on $\cH_S$ with the usual trace norm denoted $\|\cdot\|_1$. Then, $(\cJ(\cH_S),\|\cdot\|_1)$ is a Banach space. The system Hamiltonian is
\[
H_\cS\equiv\eps_0 N = \eps_ 0a^* a,
\]
where $a^*$, $a$ are the creation/annihilation operators on $\cH_\cS$ satisfying the commutation relation $[a,a^*]=\one$, $N$ is the number operator and $\eps_0>0$ is the frequency of the considered mode. We denote by $(|n\ket)_{n\in\N}$ the Fock basis. The vectors $|n\ket$ are the eigenstates of the number operator $N$, $N|n\ket=n|n\ket$, hence of $H_\cS$, and the creation/annihilation operators act on them as
\[
a|n\ket = \sqrt{n}|n-1\ket, \quad a^*|n\ket=\sqrt{n+1}|n+1\ket.
\]
For an introduction to Fock spaces we refer the reader e.g. to \cite{AJP06,BR97,DG13}.

The Hilbert space for a single atom is $$\cH_\at \equiv \C^2.$$ Each of the atoms states are $2\times 2$ density matrices. The Hamiltonian of a single atom is
\[
H_\at\equiv \eps b^*b,
\]
where $b^*$, $b$ denote the creation/annihilation operators on $\cH_\at$, i.e. $b=\left(\begin{matrix} 0 & 1 \\ 0 & 0 \end{matrix}\right)$, and $\eps>0$ is the Bohr frequency between the atomic levels (we have set the lowest energy to $0$ since it amounts to an irrelevant energy shift). Stationary states of the atom can be parametrized by the inverse temperature $\beta\in\R$ and are given by the density matrices $\rho_\at:={\e^{-\beta H_\at}}/{\tr \ \e^{-\beta H_\at}}$. We will further denote by $|-\ket$, resp. $|+\ket$, the ground, resp. excited state of an atom, namely $H_\at|+\rangle=\eps|+\rangle$ and $H_\at|-\rangle=0$.

In the dipole and rotating wave approximation the interaction is given by 
\begin{equation}\label{eq:interactionham}
H_{\rm int}\equiv \lambda(a^*\otimes b+a\otimes b^*),
\end{equation}
where $\lambda$ is (half of) the Rabi frequency. This leads to the Jaynes-Cummings Hamiltonian, see e.g. \cite{Ba99,CDG92,Du05},
\[
H\equiv H_\cS\otimes \one_\at+\one_\cS\otimes H_\at+ H_{\rm int}.
\]

The cavity will interact in a successive way with independent atoms, each for a duration time $\tau>0$. If the cavity is in a state $\rho$ before an interaction, the state of the joint cavity+atom system after interaction is therefore
\[
\e^{-i\tau H} \, \rho\otimes \rho_\at \, \e^{i\tau H}.
\]
To obtain the state of the cavity field one has to take the partial trace over the atomic part, and it is thus given by 
\begin{equation}
\label{def:rdm}
\cL(\rho):= \tr_{\cH_\at}\left[ \e^{-i\tau H} (\rho\otimes \rho_\at)\, \e^{i\tau H} \right].
\end{equation}
The map $\cL$ is the completely positive trace preserving (CPTP) map acting on $\cJ_1(\cH_\cS)$ describing the evolution of the cavity field during one interaction (discarding any measurement outcome). Such CPTP maps admit Kraus decompostions \cite{Kr83}. An explicit computation, relying on the fact that the Jaynes-Cummings hamiltonian commutes with the total number operator $N\otimes \one_\at + \one_\cS\otimes b^*b$, leads to the following Kraus representation of the map $\cL$:
\begin{equation*}\label{eq:krausrepresentation}
\cL(\rho)= \sum_{\sigma,\sigma' = \pm} V_\ss \rho V_\ss^*,
\end{equation*}
where the operators $V_\ss$ are given by 
\begin{equation}\label{eq:krausoperators}
\begin{array}{ll}
\ds V_{--}=p_\at(-)^{1/2}\,\e^{-i\tau\eps N}\,C(N),& 
\ds V_{-+}=p_\at(-)^{1/2}\,\e^{-i\tau\eps N}\,S(N+1)\,a, \\[16pt]
\ds V_{+-}=p_\at(+)^{1/2}\,\e^{-i\tau\eps N}\,S(N)\,a^*, & 
\ds V_{++}=p_\at(+)^{1/2}\,\e^{-i\tau\eps N}\,C(N+1)^*,
\end{array}
\end{equation}
with
\begin{equation*}\label{def:CSfunctions}
C(N):=\cos(\pi\sqrt{\xi N+\eta})+i\eta^{1/2}\,\frac{\sin(\pi\sqrt{\xi N+\eta})}{\sqrt{\xi N+\eta}},\quad
S(N):=\xi^{1/2}\,\frac{\sin(\pi\sqrt{\xi N+\eta})}{\sqrt{\xi N+\eta}}.
\end{equation*}
Here $\eta$ and $\xi$ are respectively the dimensionless detuning parameter and coupling constant
$$
\eta\equiv \left(\frac{\Delta\tau}{2\pi}\right)^2,\qquad \xi\equiv\left(\frac{\lambda\tau}{\pi}\right)^2,
$$
with $\Delta=\eps-\eps_0$ the detuning parameter, and 
\begin{equation*}
p_\at(-)= \frac{1}{1+\e^{-\beta\eps}}  , \quad p_\at(+) =\frac{\e^{-\beta\eps}}{1+\e^{-\beta\eps}},
\end{equation*}
are the probabilities that the atom enters in the cavity in its ground, resp. excited, state. 

One can easily verify that 
\begin{equation}\label{eq:tracepreserving}
\sum_{\sigma,\sigma' = \pm} V_\ss^*V_\ss=\one,
\end{equation} 
which traduces the trace preserving property of the map $\cL$.

As we mentioned in the introduction, quantum trajectories describe the evolution of quantum systems undergoing indirect measurements. We will come back to it in the next section but, even without further going to quantum trajectories, such an indirect measurement sheds a light on the above Kraus decomposition. If we perform a two-time measurement on the atom along the basis $\{ |-\ket,|+\ket \}$, one before and one after its interaction with the cavity, this yields two results. The first measurement gives a result $\sigma\in\{-,+\}$ with probability $p_\at(\sigma)$, the atom being then in the state $|\sigma\ket\bra \sigma|$.  If the cavity is in the state $\rho$ before the interaction, the second measurement gives then a result $\sigma'\in\{-,+\}$ with (conditional) probability
\[
P(\sigma' | \sigma) = \tr \left( \one_\cS \otimes |\sigma'\ket\bra \sigma '| \, \e^{-i\tau H} \rho \otimes |\sigma\ket\bra \sigma| \e^{i\tau H}   \right). 
\]
Altogether, the pair of results $(\sigma,\sigma')$ occurs with probability
\[
P(\sigma,\sigma') = p_\at(\sigma)\times \tr \left( \one_\cS \otimes |\sigma'\ket\bra \sigma '| \, \e^{-i\tau H} \rho \otimes |\sigma\ket\bra \sigma| \e^{i\tau H}   \right). 
\]
Introducing the operators $V_\ss:= p_\at(\sigma)^{1/2} (\one_\cS \otimes \bra \sigma '|) \, \e^{-i\tau H} \, (\one_\cS \otimes |\sigma\ket)$ we may rewrite it as
\[
P(\sigma,\sigma') = \tr\left( V_\ss \rho V_\ss^*\right),
\]
and the state of the cavity is then, after the two measurements,
\begin{equation}\label{eq:krauselement}
\rho(\sigma,\sigma') = \frac{V_\ss \rho V_\ss^*}{\tr\left( V_\ss \rho V_\ss^*\right)}.
\end{equation}
Note that the condition~(\ref{eq:tracepreserving}) ensures this defines a probability distribution on $\{-,+\}^2$. It is then easy to see that $V_\ss$ is indeed given by~(\ref{eq:krausoperators}), in other words $V_\ss$ is associated to the atomic transition from state $\sigma$ to state $\sigma'$, and $\cL(\rho)$ is then nothing but the expectation value 
\[
\cL(\rho) = \sum_{\sigma,\sigma' = \pm} \rho(\sigma,\sigma')P(\sigma,\sigma').
\]

As mentioned in the introduction, an essential feature of the dynamics generated by the Jaynes-Cummings Hamiltonian are Rabi oscillations. In the presence of $n$ photons, the probability for the atom to make a transition from its ground state to its excited state is a periodic function of time. The circular frequency of this oscillation is given by $\nu_n:=\sqrt{4\lambda^2 n+\Delta^2}$. These oscillations are at the origin of what was called a \emph{Rabi resonance} in \cite{BP09}. Such a resonance occurs when the interaction time $\tau$ is an integer multiple of the period of a Rabi oscillation, i.e. $\tau=k\frac{2\pi}{\nu_n}$ for some $k\in\N$. A positive integer $n$ is a Rabi resonance iff
\begin{equation}
\xi n+\eta=k^2,
\label{def:rabi}
\end{equation}
for some positive integer $k$. Depending on the arithmetic properties of $\eta$ and $\xi$ one easily proves (\cite{BP09}, Lemma 3.2) that the system has either no, one or infinitely many Rabi resonances. Accordingly, the system is called non-resonant, simply resonant or fully resonant. The latter two correspond to situations where one, respectively infinitely many, splitting of the cavity occurs.

\begin{definition} We say that the non-resonant condition holds if there is no Rabi resonance.
\end{definition}
In this paper we shall mostly be concerned with the non-resonant situation. Except in Section~\ref{sec:resonant} we shall assume that the non-resonant condition holds without further mentioning. 

The large time analysis of this maser model, without measurements, has been studied in \cite{BP09,Bru14}. The main result is the following
\begin{theorem}[\cite{Bru14}, Theorem 3.4]\label{thm:Lmixing} Suppose the non-resonant condition holds and $\beta>0$. Then $\rho_\inv:= \frac{\e^{-\beta\eps N}}{\tr\left(\e^{-\beta\eps N} \right)}$ is the unique invariant state of $\cL$ and for any initial state $\rho$ one has
\begin{equation}\label{eq:strongmixing}
\lim_{t\to\infty} \left\|\cL^t(\rho)-\rho_\inv\right\|_1=0.
\end{equation}
\end{theorem}
Although it is not our main concern in this paper, these results about $\cL$ will be a key step in our proofs. 

\begin{remark}\label{rem:noinvstate} 1) If $\beta\leq 0$ there is no invariant state.

2) One cannot expect any general estimate on the convergence speed in~(\ref{eq:strongmixing}), in the sense that it can be arbitrarily slow, see \cite{Bru14}. 
\end{remark}

\medskip
%%%%%%%%%%%%%%%%%%%%%%%%%%%%%%%%%%%%%%%%%%%%%%%%%%%%%%%%%%%%%%%%%%%%
%%%%%%%%%%%%%%%%%%%%%%%%%%%%%%%%%%%%%%%%%%%%%%%%%%%%%%%%%%%%%%%%%%%%

\subsection{The quantum trajectory}\label{ssec:qtraj}

The quantum trajectory we are interested in describes the evolution of the cavity interacting with a sequence of atoms, each of which is subject to a double measurement as described in Section~\ref{ssec:masermodel}. Each step gives as a result of the measurements a data in $\mathcal R=\{-,+\}^2$, and conditionally to the results of the observation the evolution of the system is updated. If at time $t$ the state of the cavity is described by $\rho_t\in \cJ_1$ and if we have collected a result $(\sigma,\sigma')$ at the $(t+1)$-st step, according to~(\ref{eq:krauselement}) the state of the cavity becomes
\[
\rho_{t+1}=\frac{V_{\sigma\sigma'}\rho_t V_{\sigma\sigma'}^*}{\tr(V_{\sigma\sigma'}\rho_t V_{\sigma\sigma'}^*)},
\]
and this occurs with probability $\tr(V_{\sigma\sigma'}\rho_t V_{\sigma\sigma'})$. As we shall see this describes recursively a Markov chain $(\rho_t)_t$ which is the main subject of our paper.

Let us make more precise the probabilistic framework allowing us to study this Markov chain. We denote by $\Omega=\mathcal R^{\mathbb N^*}$, an element of $\Omega$ shall be denoted by $\omega=(\omega_1,\omega_2,\ldots)$ where $\omega_j=(\sigma_j,\sigma_j')$ describes the result of the $j$-th measurements. On this set we introduce the cylinder algebra $\mathcal O$ generated by the cylinder sets. More precisely let $t\in\mathbb N^*$ and let $(\omega_1,\ldots,\omega_t)\in\mathcal R^t$, we denote $\Lambda_{(\omega_1,\ldots,\omega_t)}=\{z\in\Omega\,\, s.t\,\,z_1=\omega_1,\ldots,z_t=\omega_t\}$. This is an elementary cylinder of size $t$. We then define $\mathcal O_t$ the sigma-algebra generated by all the elementary cylinders of size $t$ that is
$$
\mathcal O_t=\sigma\left(\Lambda_{(\omega_1,\ldots,\omega_t)},(\omega_1,\ldots,\omega_t)\in\mathcal R^t\right).
$$
Thus we have
$$
\mathcal O=\sigma\left(\bigcup_{t\in\mathbb N}\mathcal O_t\right).
$$
This sigma-algebra describes actually the sequence of results of the entire measurement protocol. 

We introduce the random variables
$$
\begin{array}{rcc}V_i:\Omega&\rightarrow&\mathcal B(\mathcal H_S)\\
\omega&\mapsto&V_i(\omega)=V_{\omega_i}
\end{array}
$$
and 
$$
W_t(\omega)=V_t(\omega)\cdots V_1(\omega),
$$
for all $\omega\in\Omega$ and all $t\in\mathbb N^*$. The quantum trajectory is then the Markov chain, defined on the set $\cJ_1$ of density matrices over $\cH_\cS$, by
\begin{equation*}\label{def:markovchain}
\rho_{t+1}(\omega) = \frac{V_{t+1}(\omega)\rho_t(\omega) V_{t+1}(\omega)^*}{\tr(V_{t+1}(\omega)\rho_t(\omega) V_{t+1}(\omega)^*)},
\end{equation*}
or in an equivalent way, and if $\rho$ is the initial state of the cavity,
\begin{equation*}\label{def:markovchain2}
\rho_t(\omega) = \frac{W_t(\omega)\rho W_t(\omega)^*}{\tr(W_t(\omega)\rho W_t(\omega)^*)}.
\end{equation*}
Note that $\rho_t(\omega)$ and $W_t(\omega)$ actually depend only on $\omega_1,\ldots,\omega_t$, hence are $\cO_t$ measurable. Moreover, according to the previous discussion, if $\rho$ is the initial state of the cavity then the probability that one has obtained $(\omega_1,\ldots,\omega_t)$ as the first $t$ outcomes of the measurements is
\begin{equation}\label{eq:probaoutcomes}
\P(\omega_1,\ldots,\omega_t | \rho_0=\rho) = \tr\left( W_t(\omega) \rho W_t^*(\omega)\right).
\end{equation}
It follows, as mentioned in the introduction, that $\mathbb E(\rho_t|\rho_0=\rho)=\mathcal L^t(\rho)$.

If $\mu$ denotes the counting measure on $\cR$, i.e.
$$
\mu=\sum_{y\in\mathcal R}\delta_y,
$$ 
the Markov kernel associated to this Markov chain is thus
\begin{eqnarray}\label{def:markovkernel}
\Pi(\rho,S)& =&\sum_{y\in\mathcal R} \one_S\left(\frac{V_y\rho V_y^*}{\tr(V_y\rho V_y^*)}\right) \tr(V_y\rho V_y^*)\nonumber\\&=&\int_\cR  \one_S\left(\frac{V_y\rho V_y^*}{\tr(V_y\rho V_y^*)}\right) \tr(V_y\rho V_y^*)  \d\mu(y),
\end{eqnarray}
where $\rho\in\cJ_1$ and $S\in \mathfrak B$ the Borel sigma-algebra over $\cJ_1$.

Our main result concerns invariant measures for this Markov chain. We recall that if $\nu$ is a probability measure over $(\cJ_1,\mathfrak B)$ then $\nu\Pi$ is the probability measure defined by
\[
\nu\Pi(S) = \int \Pi(\rho,S) \d\nu(\rho),
\]
and that $\nu$ is $\Pi$-invariant if $\nu\Pi=\nu$. We equip the set of probability measures over $(\cJ_1,\mathfrak B)$ with the Wassertein metric of order $1$. Using the Kantorovich-Rubinstein duality theorem, for two probability measures $\nu_1, \nu_2$ it can be defined by
\begin{equation}\label{eq:def Wasserstein}
W_1(\nu_1,\nu_2)=\sup\left\{\int_{\cJ_1} f(\rho)\d\nu_1(\rho)- \int_{\cJ_1} f(\rho)\d\nu_2(\rho), f\in \operatorname{Lip}_1(\cJ_1)\right\}
\end{equation}
with
$$\operatorname{Lip}_1(\cJ_1)=\left\{f:\cJ_1\to\rr, |f(\rho)-f(\varrho)|\leq \|\rho-\varrho\|_1\right\}.$$

\begin{theorem}\label{thm:invariantmeasure} Suppose the non-resonant condition holds and $\beta>0$. Then 
\begin{equation}\label{def:invariantmeasure}
\nu_\inv \equiv \sum_{n=0}^{+\infty} \bra n, \rho_\inv n\ket \delta_{|n\ket\bra n|}
\end{equation}
is the unique $\Pi$-invariant probability measure and for any probability measure $\nu$ over $(\cJ_1,\mathfrak B)$ 
$$\lim_{t\to\infty}W_1(\nu\Pi^t, \nu_\inv)=0.$$
In particular, $(\nu\Pi^t)_{t\in \nn}$ converges weakly to $\nu_\inv$.
\end{theorem}

\begin{remark} If the non-resonant condition holds and $\beta\leq 0$ there is actually no invariant measure, see Corollary~\ref{coro:invariantmeasureinfo}.
\end{remark}

\begin{remark} The proof of the convergence uses Equation~(\ref{eq:strongmixing}). As mentioned in Remark~\ref{rem:noinvstate} the latter can be arbitrarily slow, hence one can not expect any rate of convergence in the above theorem.
\end{remark}

The strategy of proof of Theorem~\ref{thm:invariantmeasure} is similar to the one developed in \cite{BFPP19}. As already mentioned it partly relies on Equation~\eqref{eq:strongmixing}. More precisely, the latter implies that the dynamical system on $\Omega$ induced by the left-shift is strongly mixing in total variation -- see Theorem~\ref{thm:mixing-dynsys}. Then, the trajectory $(\rho_t)_t$ is consistently estimated as $t$ grows using a filter independent of the initial state. Namely, it depends only on the realization $\omega\in \Omega$. The consistency of the estimator (or filter) is proved using purification of the quantum trajectories. The proof of purification is the main innovation, as known proofs rely on compactness of the set of density operators, assumption that is not verified for the one atom maser. The next section is dedicated to the dynamical system on $\Omega$. The purification theorem is then given in Section~\ref{sec:purification-results}.

\medskip
%%%%%%%%%%%%%%%%%%%%%%%%%%%%%%%%%%%%%%%%%%%%%%%%%%%%%%%%%%%%%%%%%%%%
%%%%%%%%%%%%%%%%%%%%%%%%%%%%%%%%%%%%%%%%%%%%%%%%%%%%%%%%%%%%%%%%%%%%

\subsection{Dynamical system properties}\label{sec:dynamical sys}

Let $\theta:\Omega\to\Omega$ denote the left-shift $\theta(\omega)_t=\omega_{t+1}$. This section is devoted to a result of mixing type for the dynamical system $(\Omega,\theta)$. Note that it can be easily generalized to arbitrary repeated quantum measurements using the instrument formalism -- see \cite{DL70}.

Let us start by defining some probability measures. They are defined by Kolmogorov extension theorem. The condition of Equation~\eqref{eq:tracepreserving} ensures the theorem can be applied.
\begin{definition}\label{def:Prho}
  For any state $\rho\in \cJ_1$, let $\pp^\rho$ be the probability measure over $\Omega$ defined by
  $$\pp^\rho(O)=\int_O \tr(\rho W_t^*(\omega)W_t(\omega))\d\mu_t(\omega)$$
  for any $t\in \nn^*$ and $O\in \cO_t$.
\end{definition}
In this definition $\mu_t$ is the product measure of $\mu$ on the first $t$ terms of the sequence $\omega$ and an arbitrary probability measure on the rest of the terms. The fact that $O\in \cO_t$ ensures that the choice of this ``tail'' probability measure is not relevant.
\begin{remark} Note that the restriction of $\P^\rho$ to $\cO_t$ is nothing but the conditional probability distribution in Equation~\eqref{eq:probaoutcomes}.
\end{remark}

We can now state the two main results of this section. The first one is a strong mixing result on the measurement outcomes.
\begin{theorem}\label{thm:mixing-dynsys}
 Assume the non-resonant condition holds and $\beta>0$. Then for any state $\rho\in \cJ_1$,
 $$\lim_{t\to\infty}\|\pp^\rho\circ\theta^{-t}-\pp^{\rho_{\rm inv}}\|_{TV}=0$$
 where $\|\cdot\|_{TV}$ is the total variation norm.
\end{theorem}

The second result concerns the transfer of the notion of absolute continuity between states to their associated probability distribution. We recall that if $\rho$ and $\varrho$ are two states, $\varrho$ is said to be absolutely continuous with respect to $\rho$, denoted $\varrho \ll \rho$, if $\ker(\rho)\subset \ker(\varrho)$. In particular if $\rho$ is a faithful state then any state is absolutely continuous with respect to it.
\begin{proposition}\label{prop:absolutecontinuity} If $\rho,\varrho$ are two states such $\varrho\ll \rho$ then $\P^\varrho \ll \P^\rho$. In particular if $\rho$ is faithful then $\P^\varrho \ll \P^\rho$ for all states $\varrho$.
\end{proposition}
\begin{remark} When $\rho$ is faithful this result will be refined in the sequel by exhibiting the density of $\P^\varrho$ with respect to $\P^\rho$, see Proposition~\ref{prop:Mn-martingale}. We shall use the above proposition several time with the faithful invariant state $\rho=\rho_\inv$.
\end{remark}

Both these results are actually the consequence of an extension theorem for positive operator valued measures (POVM).
\begin{lemma}\label{lem:POVM}
There exists a POVM $P:\cO\to\cB(\cH_\cS)$ such that for any state $\rho\in \cJ_1$ and $O\in \cO$,
$$\pp^\rho(O)=\tr(\rho P(O)).$$
\end{lemma}
\begin{proof} Let $\mathcal G=\mathcal P(\mathcal R)$ be the collection of all subsets of $\cR$ (recall $\cR$ is a finite set). For any $t\in \nn^*$ let $(\mathcal R^t,\mathcal G^{\otimes t},P_t)$ be the positive operator valued measure (POVM) defined by 
$$P_t(O)=\int_O V_{\omega_1}^*\dotsb V_{\omega_t}^*V_{\omega_t}\dotsb V_{\omega_1}\d\mu^{\otimes t}(\omega_1,\dotsc,\omega_t)$$ 
for any $O\in \mathcal G^{\otimes t}$. Since for any $O\in \mathcal G^{\otimes t}$, Equation~\eqref{eq:tracepreserving} implies $P_{t+1}(O\times \cR)=P_t(O)$, the family of POVMs $(P_t)_t$ is consistent. Then, \cite[Corollary~1]{Tu08} implies there exists a POVM $(\Omega,\cO,P)$ such that for any $O\in \cO_t$, $P(O)=P_t(O)$ where we denote $O$ as an element of $\cO_t$ and as an element of $\mathcal G^{\otimes t}$ with the same letter. Since $\mathcal R$ is finite, it follows that for any $O\in \cO$, $\pp^\rho(O)=\tr(\rho P(O))$.
\end{proof}
Note that the proof relies only on the fact that Kolmogorov's extension theorem can be generalized to POVM using Riesz's representation theorem as done in \cite{Tu08}. This lemma then implies the following Lipschitz continuity of the family of probability measures $(\P^\rho)_\rho$.

\begin{proposition}\label{dist_var_totale} For all states $\rho,\varrho\in \cJ_1$, the following holds
$$
\|\mathbb P^\rho-\mathbb P^\varrho\|_{TV}\leq\|\rho-\varrho\|_{1}.
$$
\end{proposition} 

This result has been given in \cite{BFP23} in the finite dimensional setting. The proof is the same. Since it is rather short we give it for the reader convenience.
\begin{proof} Lemma~\ref{lem:POVM} ensures there exists a POVM $P$ such that for any $O\in\cO$
  $$\pp^\rho(O)-\pp^\varrho(O)=\tr[(\rho-\varrho) P(O)].$$
  Then H\"older inequality for Schatten norms implies
\begin{eqnarray*}
|\mathbb{P}^{\rho}(O) - \mathbb{P}^{\varrho}(O)| & \leq & \|\rho-\varrho\|_1 \left\| P(O) \right\| \\
 & \leq & \|\rho-\varrho\|_1,
\end{eqnarray*}
where $\|\cdot  \|$ denotes the usual operator norm and we have used $P(O)\leq \id_\cH$ in the last line.
\end{proof}

We can now prove the two results of this section. 
\begin{proof}[Proof of Theorem~\ref{thm:mixing-dynsys}]
  By definition, for any $O\in \cO$, $\pp^{\rho}\circ\theta^{-t}(O)=\mathbb P^\rho(\mathcal R^t\times O)$. Then Lemma~\ref{lem:POVM} leads to
  $$\pp^{\rho}\circ\theta^{-t}(O)=\sum_{(\omega_1,\dotsc,\omega_t)\in \cR^t}\tr(\rho V_{\omega_1}^*\dotsb V_{\omega_t}^*P(O)V_{\omega_t}\dotsb V_{\omega_1}).$$
   Using the cyclicity of the trace and the definition of $\mathcal L$ we further have
  $$\pp^{\rho}\circ\theta^{-t}(O)=\tr(\mathcal L^t(\rho)P(O))=\pp^{\mathcal L^t(\rho)}(O)$$
  and Proposition~\ref{dist_var_totale} together with Theorem~\ref{thm:Lmixing} yield the result.
\end{proof}

\begin{proof}[Proof of Proposition~\ref{prop:absolutecontinuity}]
 Let $\rho,\varrho$ be two states such that $\varrho\ll\rho$, i.e. $\ker(\rho)\subset \ker(\varrho)$. If $O\in\cO$ is such that $\P^\rho(O)=0$ then, using Lemma~\ref{lem:POVM} we have $\operatorname*{range}(P(O))\subset \ker\rho$. Hence $\operatorname*{range}(P(O))\subset \ker\varrho$ and another use of Lemma~\ref{lem:POVM} leads to $\pp^\varrho(O)=\tr(\varrho P(O))=0$ and indeed $\pp^\varrho\ll \pp^\rho$.
\end{proof}

\medskip
%%%%%%%%%%%%%%%%%%%%%%%%%%%%%%%%%%%%%%%%%%%%%%%%%%%%%%%%%%%%%%%%%%%%
%%%%%%%%%%%%%%%%%%%%%%%%%%%%%%%%%%%%%%%%%%%%%%%%%%%%%%%%%%%%%%%%%%%%

\subsection{Purification}\label{sec:purification-results}

It is easy to see that the set of pure states is preserved along the trajectory. Namely, if the initial state $\rho$ is a pure state, i.e. $\rho=|\phi\ket\bra \phi|$ for some unit vector $\phi\in\cH_\cS$, then $\rho_t(\omega)$ remains a pure state for all $t\in\N$ and $\omega\in\Omega$. Even stronger, we shall see in Section~\ref{ssec:classicalchain} that if this unit vector is a Fock vector $|n\ket$ then $\rho_t$ remains a Fock state. As the result of Theorem~\ref{thm:invariantmeasure} suggests, one can actually expect that starting from any state $\rho$ the evolved state $\rho_t(\omega)$ approaches a pure state, and moreover that the latter is a Fock state. When $\cH_\cS$ is finite dimensional this phenomenon is known as purification. The first proof of purification, under some suitable and natural assumption, was published in \cite{MK06}. We refer the reader to e.g. \cite{BFPP19} and references therein for alternative proofs and more details. Our second main result, which is also the second main ingredient in the proof of Theorem~\ref{thm:invariantmeasure}, is a precise formulation of the above heuristic.

One can infer from Equation~\eqref{eq:krausoperators} that, up to normalization, the operators $V_\ss$ map Fock states to Fock states, see Equation~\eqref{eq:krausonfock}, and more precisely (except $V_{-+}|0\ket$ which vanishes) map the state $|k\ket$ to either $|k-1\ket$, $|k\ket$ or $|k+1\ket$ depending on $(\sigma,\sigma')$. This is not surprising because of the form of the interaction in the Jaynes-Cummings hamiltonian, see Equation~(\ref{eq:interactionham}). This will thus lead to a natural birth and death Markov chain, see Section~\ref{ssec:classicalchain} for more details. 

We actually construct explicitly a process on $\nn\cup\{-1\}$ that we use to obtain the estimator of the quantum trajectories mentioned at the end of Section~\ref{ssec:qtraj}. To construct this process we introduce the following quantities.
\begin{definition}\label{def:shift} For all $t\in\N^*$, define $s_t:\Omega\to \Z$ as,
\[
s_t(\omega):=\sum_{u=1}^t \delta(\omega_u) \quad \mbox{where} \quad 
\delta(\omega_u) = \left\{\begin{array}{cl} 1 & \mbox{if } \omega_u=(+,-), \\ -1 & \mbox{if } \omega_u=(-,+), \\ 0 & \mbox{otherwise}.\end{array}\right.
\]
The quantity $s_t(\omega)$ amounts for the total left or right shift induced by $W_t(\omega)$ on Fock states.
\end{definition}

\begin{definition}\label{def:extinctiontime} For all $k\in\N^*$, define the extinction times $\tau^k:\Omega\to \N^*\cup\{+\infty\}$ by 
\begin{equation}\label{eq:extinctiontime}
\tau^k(\omega) =\inf\{t\geq 1: W_t|k\ket=0\}
\end{equation}
for any $\omega\in\Omega$. In the non resonant case this is equivalent to $\tau^k= \inf\{t\geq 1: k+s_t(\omega)<0\}.$
\end{definition}

\begin{definition}\label{def:evolvedfockstate} For all $t\in\N^*$, define $N_t:\N\times \Omega\to \N\cup\{-1\}$ as follows. For $\omega\in\Omega$: 
$$N_t(k,\omega)=\left\{ \begin{array}{cl} k+s_t(\omega) & \mbox{if } t<\tau^k(\omega),\\ -1 & \mbox{otherwise.}\end{array}\right.$$
This defines a family of $\mathcal O$ measurable random variables that we shall denote $(N_t(k,.))_{k\in\mathbb N,t\in\mathbb N^*}.$
\end{definition}

\begin{remark}\label{rk:W_t, N_t} If the system is initially in the pure Fock state $|k\ket\bra k|$, then $|N_t(k,\omega)\ket\bra N_t(k,\omega)|$ is the pure Fock state which is reached at time $t$ if we have observed $\omega$ (see Equation~(\ref{eq:evolvedfockstate})), i.e. it is such that $$\ds |N_t(k,\omega)\ket = \frac{W_t(\omega)|k\ket}{\left\| W_t(\omega)|k\ket\right\|},$$ 
if $t<\tau^k(\omega)$.

The state $-1$ plays the role of a cemetery state as if the extinction of the birth and death process occurs and in that case $W_t(\omega)|k\ket=0$. For all $t$ and $k$, the map $N_t(k,\cdot)$ is clearly $\cO_t$ measurable and then $\mathcal O$ mesurable.

As we shall see, the state $0$ of the involved birth and death process is not absorbing (one can jump from $0$ to $1$ with a positive probability). The state $0$ of the birth and death process corresponds actually to the Fock state $\vert 0\ket$. In particular it follows directly from Definition~\ref{def:Prho} that $\P^{|k\ket\bra k|}(\tau^k<\infty)=0$ for all $k$ i.e. starting from a Fock state $k$ we always have $\tau^k=\infty$ almost surely. In the sequel we shall consider $\tau^k$ when $k$ is not the starting point and in this case one can have $\tau^k<\infty$ with non zero probability.
\end{remark}

We can now state our purification theorem.
\begin{theorem}\label{thm:purification} Suppose the non-resonant condition holds and $\beta>0$. Then there exists a random variable $n_\infty$ valued in $\mathbb N$ such that
\begin{enumerate}
\item for any $\rho\in \cJ_1$, the law of $n_\infty$ under $\P^{\rho}$ is given by $\P^{\rho}(n_\infty=n)=\bra n,\rho n\ket$. In particular $\E^{\rho_\inv}\left( |n_\infty\ket\bra n_\infty|\right) =\rho_\inv$,
\item $\tau^{n_\infty}=+\infty$, $\P^{\rho_\inv}$-almost surely,
\item For any initial state $\rho\in\mathcal J_1$ the associated quantum trajectory
$$
\rho_t=\frac{W_t\rho W_t^*}{\tr(W_t\rho W_t^*)}
$$
satisfies, $\P^\rho$-almost surely and in $L^1$, 
\begin{equation}\label{eq:purification}
\lim_{t\to+\infty} \left\| \rho_t - \left|N_t(n_\infty,.)\right\ket\left\bra N_t(n_\infty,.)\right| \right\|_1 =0.
\end{equation}
\end{enumerate}
\end{theorem}

The meaning  of Equation~(\ref{eq:purification}) is that the state $\rho_t$ indeed approaches a pure state which is a Fock state. More precisely, independently of the initial state of $\cS$ the quantum trajectory asymptotically approaches the one which starts from the (random) Fock state $|n_\infty(\omega)\ket\bra n_\infty(\omega)|$. This evolved Fock state $\left|N_t(n_\infty,.)\right\ket\left\bra N_t(n_\infty,.)\right|$ will be our estimator in the proof of Theorem~\ref{thm:invariantmeasure}. Note also that the law of $n_\infty$, given the initial state $\rho$, is actually the Born's rule law of a measurement of the number operator $N$ with respect to the same state $\rho$. Finally, remark that point (2) means that $\P^{\rho_\inv}$-almost surely extinction does not occur when one starts from $n_\infty$. Because $\rho_\inv$ is faithful, and due to Proposition~\ref{prop:absolutecontinuity}, it is also true $\P^\rho$-almost surely for any state $\rho.$

%%%%%%%%%%%%%%%%%%%%%%%%%%%%%%%%%%%%%%%%%%%%%%%%%%%%%%%%%%%%%%%%%%%%
%%%%%%%%%%%%%%%%%%%%%%%%%%%%%%%%%%%%%%%%%%%%%%%%%%%%%%%%%%%%%%%%%%%%
%%%%%%%%%%%%%%%%%%%%%%%%%%%%%%%%%%%%%%%%%%%%%%%%%%%%%%%%%%%%%%%%%%%%
%%%%%%%%%%%%%%%%%%%%%%%%%%%%%%%%%%%%%%%%%%%%%%%%%%%%%%%%%%%%%%%%%%%%

\section{Purification of the trajectory}\label{sec:purification}

%%%%%%%%%%%%%%%%%%%%%%%%%%%%%%%%%%%%%%%%%%%%%%%%%%%%%%%%%%%%%%%%%%%%
%%%%%%%%%%%%%%%%%%%%%%%%%%%%%%%%%%%%%%%%%%%%%%%%%%%%%%%%%%%%%%%%%%%%

\subsection{The quantum trajectory on Fock states}\label{ssec:classicalchain}

Using Equation~(\ref{eq:krausoperators}) a simple computation gives
\begin{equation}\label{eq:krausonfock}
\begin{array}{ll}
\ds V_{--}|k\ket\bra k| V_{--}^* = p_\at(-)\,\left[1-\alpha_k\right] \, |k\ket\bra k|, &
\ds V_{-+}|k\ket\bra k| V_{-+}^* = p_\at(-)\,\alpha_k \, |k-1\ket\bra k-1|, \\[16pt]
\ds V_{+-}|k\ket\bra k| V_{+-}^* = p_\at(+)\,\alpha_{k+1} \, |k+1\ket\bra k+1|, &
\ds V_{++}|k\ket\bra k| V_{++}^* = p_\at(+)\,\left[1-\alpha_{k+1}\right] \, |k\ket\bra k|,
\end{array}
\end{equation}
where $\ds \alpha_k:= \sin^2(\pi\sqrt{\xi k+\eta})\frac{\xi k}{\xi k+\eta}\in[0,1]$. Hence the set $\left\{ |k\ket\bra k|, \ k\in\N\right\}$ is invariant along the trajectory. More precisely,  note that $V_{-+}|0\ket$ vanishes and otherwise
\[
\frac{V_\ss |k\ket}{\|V_\ss |k\ket\|} =  |k+\delta(\sigma,\sigma')\ket,
\]
where $\delta(\sigma,\sigma')$ is the one-step shift introduced in Definition~\ref{def:shift}.

If we identify a Fock state $|k\ket\bra k|$ with the integer $k$, the quantum trajectory is thus identified with the classical Markov chain on $\mathbb N$ whose transition matrix $P$ is given by 
$$p(0,0)=1-p_\at(+)\alpha_{1},\quad p(0,1)=p_\at(+)\alpha_{1},$$
and for $k\in\mathbb N^*$ by
\begin{eqnarray}
p(k, k-1) & = & p_\at(-) \alpha_k, \label{eq:transitionproba-}\\ 
p(k, k+1)  & = & p_\at(+) \alpha_{k+1}, \label{eq:transitionproba+}\\ 
p(k, k)  & = & 1- p(k,k-1)-p(k,k+1),\label{eq:transitionproba0}
\end{eqnarray}
and $p(k,\ell)=0$ otherwise. The quantity $p(k, \ell)$ denotes the probability that a transition from exactly $k$ to $\ell$ photons occurs in the cavity. Note that $p(k,k-1)$ and $p(k-1,k)$ vanish exactly when $k$ is a Rabi resonance. Hence under the non-resonant condition they never vanish. This is exactly the context of a classical birth and death processes. As we shall see, this classical Markov chain will be central in our final result.

Classical results about birth and death Markov chains -- see e.g. \cite{Bre13} -- lead to the following proposition.

\begin{proposition}\label{prop:classicalmarkov}
The Markov chain on $\nn$ with transition kernel $P=(p(k,\ell))_{k,\ell\in \nn}$ is irreducible and positive recurrent. Its unique invariant probability measure is
$$\mu_{Gibbs}(k)=\frac{e^{-k\beta\eps}}{1+e^{-\beta\eps}} = \bra k,\rho_\inv k\ket.$$
Furthermore the Markov chain is aperiodic and for any probability measure $\mu$ over $\nn$,  $\lim_t \mu P^t= \mu_{Gibbs}$ weakly.
\end{proposition}

\begin{proof} An immediate consequence of the non-resonant condition is that for any $k,\ell\in \nn$, there exists $t\in \nn$ such that $P^t(k,\ell)>0$. A direct computation shows that the (not normalized) measure $\mu$ defined for any $k\in\nn$ by
$$\mu(k)=\prod_{j=1}^k\frac{p(j-1,j)}{p(j,j-1)}=\left(\frac{p_\at(+)}{p_\at(-)}\right)^k=e^{-k\beta\eps},$$
is invariant. It is finite therefore the Markov chain is positive recurrent. Irreducibility yields $\mu_{Gibbs}=\mu/\mu(\nn)$ is the unique invariant probability measure. Since $\beta>0$ and $\alpha_1\leq 1$, $p(0,0)>0$ therefore the Markov chain is aperiodic. The convergence follows.
\end{proof}

\begin{remark} 
This proposition can also be seen as a direct consequence of \cite[Theorem~3.4]{Bru14}. For $\rho\in\cJ_1$ let $\mu_\rho$ be the probability measure defined by $\mu_\rho(n)=\langle n,\rho n\rangle$. Then $\langle n,\cL^t(\rho)n\rangle=\mu_\rho P^t$ and the irreducibility and positive recurrence follow from $\lim_t\|\cL^t(\rho)-\rho_{inv}\|_1=0$ remarking that $\mu_{\rho_{inv}}=\mu_{Gibbs}$.
\end{remark}

Equations~(\ref{eq:krausonfock}) together with Definitions~\ref{def:shift},~\ref{def:extinctiontime} and~\ref{def:evolvedfockstate} imply that the above defined birth and death Markov chains encodes the action of $(W_t)_t$ on the set of Fock states. Namely,
\begin{itemize}
\item[i)] Recall the extinction times $\tau^k$ introduced in Definition~\ref{def:extinctiontime}. As already mentioned, for all $k$ we have $\pp^{|k\ket\bra k|} (\tau^k<+\infty)=0$.
\item[ii)] if $\tau^k(\omega)=+\infty$ then, for any $t$, 
\begin{equation}\label{eq:evolvedfockstate}
\frac{W_t(\omega)|k\ket\bra k| W_t^*(\omega) }{\tr\left(W_t(\omega)|k\ket\bra k|W_t^*(\omega)\right)} = |N_t(k,\omega)\ket\bra N_t(k,\omega)|,
\end{equation}
\item[iii)] if $\tau^k(\omega)<+\infty$ then Equation~\eqref{eq:evolvedfockstate} holds for $t<\tau^k(\omega)$. For $t\geq \tau^k(\omega)$, from Definition~\ref{def:evolvedfockstate}, $N_t(k,\omega)=-1$ and we denote by $|-1\ket$ an arbitrary unit vector in $\cH_\cS$ which is not in $\{|k\ket, \ k\in\N\}$.
\end{itemize}

\noindent The following lemma follows immediately from Definition~\ref{def:shift}.
\begin{lemma}\label{lem:translationinvariance} For $k,m\in\mathbb N$ we have $\{\tau^{k}=+\infty\}\subset \{\tau^{k+m}=+\infty\}$ and, on $\{\tau^{k}=+\infty\}$, 
\[
N_t(k+m,.)=N_t(k,.)+m, \quad \forall t\in \nn.
\]
\end{lemma}

We end this section with the consequence of the recurrence property on the quantum trajectory. 
\begin{proposition}\label{prop:recurrenceqtraj} For any $k\in\N$ let  $\ds A_k:= \{\omega\, |\, N_t(k,\omega)=k \mbox{ for infinitely many } t\in\nn\}$,
and let $\ds A=\cup_k A_k$. Then $\P^{\rho_\inv}(A)=1$. Moreover, for any $\omega\in A$ there exists a strictly increasing sequence of integers $(t_j)_j$ such that for all $k$
\begin{itemize}
\item[$\bullet$] if $\tau^k(\omega)=+\infty$ then $\ds N_{t_j}(k,\omega) = k$ for all $j$,
\item[$\bullet$] if $\tau^k(\omega)<+\infty$ then $\ds N_{t_j}(k,\omega) = -1$ for any $j$ large enough.
\end{itemize}
\end{proposition}

\begin{proof} Let $k\in\mathbb N$. From Proposition~\ref{prop:classicalmarkov}, the recurrence property with $\mu=\delta_k$ means that $\ds \P^{|k\ket\bra k|} (A_k)=1$. Hence $\ds \P^{|k\ket\bra k|} (A)=1$ for all $k$. Since $\ds \rho_\inv = \sum_k \bra k,\rho_\inv k\ket |k\ket\bra k|$, using Definition \eqref{def:Prho}, we get
\[
\P^{\rho_\inv}(A)=\sum_k\bra k,\rho_\inv k\ket\P^{|k\ket\bra k|}(A)=1.
\]
Take now $\omega\in A$ and choose $l$ such that $\omega\in A_l$. By definition, there exists a strictly increasing subsequence $(t_j)_j$ such that $N_{t_j}(l,\omega)= l$ for all $j$. In particular $\tau^l(\omega)=+\infty$. Pick now $k\in\N$. If $\tau^k(\omega)<+\infty$, then $N_t(k,\omega)=-1$ for all $t\geq \tau^k$ and the proposition holds, while if $\tau^k(\omega)=+\infty$ the result follows from Lemma~\ref{lem:translationinvariance}.
\end{proof}

\medskip
%%%%%%%%%%%%%%%%%%%%%%%%%%%%%%%%%%%%%%%%%%%%%%%%%%%%%%%%%%%%%%%%%%%%
%%%%%%%%%%%%%%%%%%%%%%%%%%%%%%%%%%%%%%%%%%%%%%%%%%%%%%%%%%%%%%%%%%%%

\subsection{A key martingale}\label{sec:martigale}

This section is devoted to the analysis of a key martingale. While inspired by the finite dimensional situation studied in \cite{BFPP19}, the fact that $\cH_\cS$ has infinite dimension imposes a more careful approach. Some arguments are based on the theory of Banach space valued martingales. On this subject we refer the reader to e.g. \cite{Pi16}.

\begin{definition} For all $t\in\N^*$, let
\begin{equation}\label{def:Mn}
M_t(\omega)=\frac{\rho_\inv^{\frac12}W_t^*(\omega)W_t(\omega)\rho_\inv^{\frac12}}{\tr(W_t^*(\omega)W_t(\omega)\rho_\inv)}.
\end{equation}
\end{definition}

\begin{remark} The state $\rho_\inv$ is a faithful state so $M_t(\omega)$ is a well defined state for any $\omega$.
\end{remark}

The random operator $M_t$ is related to the polar decomposition of $W_t$. Indeed, up to the normalization factor $\sqrt{\tr(W_t^*(\omega)W_t(\omega)\rho_\inv)}$, $M_t^{1/2}$ is the positive part of the polar decomposition of $W_t\rho_\inv^{1/2}$. As such it inherits the following elementary, but important, property.
\begin{lemma}\label{lem:Mndiagonal} For any $t\in \nn$ and $\omega\in \Omega$, the operator $W_t^*(\omega)W_t(\omega)$ is a bounded function of the number operator $N$. As a consequence so is $M_t(\omega)$.
\end{lemma}

\begin{proof} It follows from Equation \eqref{eq:krausoperators} that for any $\sigma,\sigma'\in \{+,-\}$ the operators $V_\ss^*V_\ss$ are bounded functions of $N$. Moreover for any bounded function $f:\nn\to \cc$,  $af(N)=f(N+1)a$. An induction argument then proves that $W_t^*(\omega)W_t(\omega)$ is a function of $N$. 
Since $\rho_\inv$ is also a function of $N$, see Theorem~\ref{thm:Lmixing}, it follows that $M_t(\omega)$ is a function of $N$.
\end{proof}

As an immediate consequence we have the following decomposition result of the measures $\P^\rho$ which we shall use several times in the sequel. We have already used such a decomposition in the particular case of $\P^{\rho_\inv}$ in the proof of Proposition~\ref{prop:recurrenceqtraj}.
\begin{lemma}\label{lem:Prho decomposition} For any $\rho \in\cJ_1$ we have
\[
\P^\rho = \sum_{n\in\N} \bra n, \rho\, n\ket \P^{|n\ket\bra n|}.
\]
\end{lemma}

\begin{proof} It suffices to prove equality when applied on elementary cylinder sets $\Lambda_{(\omega_1,\ldots,\omega_t)}$. We then have
\begin{eqnarray*}
\P^\rho (\Lambda_{(\omega_1,\ldots,\omega_t)}) & = & \tr(\rho W_t^*(\omega)W_t(\omega)) \\
 & = & \sum_{n\in\N} \bra n, \rho W_t^*(\omega)W_t(\omega) \, n\ket \\
 & = &  \sum_{n\in\N} \bra n, \rho \, n\ket \bra n, W_t^*(\omega)W_t(\omega) \, n\ket  \\
 & = &  \sum_{n\in\N} \bra n, \rho\, n\ket  \P^{|n\ket\bra n|}  (\Lambda_{(\omega_1,\ldots,\omega_t)}), 
\end{eqnarray*}
where we have used that $W_t^*W_t$ is a function of the number operator in the third line.
\end{proof}

Before turning to the analysis of $(M_t)_t$ let us comment on its meaning. The Fock states are eigenstates of $\rho_\inv$ hence, for all $t\in\N^*$ and $n\in\N$, one has %$\P^{\rho}$-almost surely
\[
\bra n, \rho_\inv^{\frac12}W_t^*W_t\rho_\inv^{\frac12} n\ket = \bra n,\rho_\inv n\ket \times \tr\left( W_t^*W_t |n\ket\bra n| \right).
\]
On the right-hand side, the second factor is precisely the probability that we observe the first $t$ outcomes $(\omega_1,\ldots,\omega_t)$ if the initial state is $|n\ket\bra n|$, while the first one can also be seen as the probability that we obtain $|n\ket\bra n|$ if the initial state is distributed according to the invariant measure $\nu_\inv$, see Equation~\eqref{def:invariantmeasure}. If on $\mathcal J_1 \times \Omega$ we consider the sigma algebras
\begin{equation}\label{def:extendedalgebra}
\cF_t=\mathfrak B\otimes\mathcal O_t,\quad \cF=\mathfrak B\otimes \mathcal O,
\end{equation}
where we recall that $\mathfrak B$ is the Borel sigma-algebra on $\cJ_1$, and introduce
\begin{equation}\label{def:probaMinterpretation}
\P(S,O_t) := \int_{S\times O_t} \tr\left(W_t(\omega)\rho W_t^*(\omega)\right)\,\d\nu_\inv(\rho)\d\mu^{\otimes t}(\omega),
\end{equation}
then $\P$ defines a probability measure on $(\cJ_1\times\Omega,\cF)$ and we can write
\begin{equation}\label{eq:probastateoutcomeinvariantmeasure}
\bra n, \rho_\inv^{\frac12}W_t^*(\omega)W_t(\omega)\rho_\inv^{\frac12} n\ket = \P\left(\{|n\ket\bra n|\}, \Lambda_{(\omega_1,\ldots,\omega_t)}\right).
\end{equation}
On the other hand,
\begin{equation}\label{eq:probaoutcomeinvariantmeasure}
\tr\left( W_t^*(\omega)W_t(\omega) \rho_\inv \right) = \P\left(\cJ_1, \Lambda_{(\omega_1,\ldots,\omega_t)}\right).
\end{equation}
With a slight abuse of notation we can therefore write
\begin{equation}\label{eq:Mtinterpretation}
\bra n, M_t(\omega) n\ket = \frac{\P(n;\omega_1,\ldots,\omega_t)}{\P(\omega_1,\ldots,\omega_t)} = \P(n|\omega_1,\ldots,\omega_t).
\end{equation}
In other words $\bra n, M_t(\omega) n\ket$ can be interpreted as the probability to start in the Fock state $|n\ket\bra n|$ if the initial state is distributed according to the invariant distribution $\nu_\inv$ and we have observed $(\omega_1,\ldots,\omega_t)$ as the $t$ first measurement outcomes.

\begin{remark} The probability measure $\P$ is a special case, when $\nu=\nu_\inv$, of the family $\P_\nu$ we will consider in Section~\ref{sec:invariantmeasure}.% Indeed as we shall see  $\P=\P_{\nu_\inv}$. 
\end{remark}

\begin{lemma}\label{lem:Mn-martingale} The process $(M_t)_t$ is a $\P^{\rho_\inv}$ uniformly bounded martingale in $\cJ_1$. As a consequence there exists $M_\infty$ such that $(M_t)_t$ converges in trace norm, $\mathbb P^{\rho_\inv}$-almost surely and in $L^1$, to $M_\infty$.
\end{lemma}

\begin{proof} The main point is to prove the martingale property. The proof is an immediate adaptation of the one in the finite dimensional case, see e.g. \cite{BFPP19}. We provide it for the reader convenience.

Let $O_t\in \cO_t$, then it follows from Definition~\ref{def:Prho} and Equation~(\ref{def:Mn}) that
\begin{eqnarray*}
\mathbb E_{\rho_\inv}[M_{t+1}\one_{O_t}]&=&\int_{O_t\times \cR} M_{t+1} \d \P^{\rho_\inv}\\ & = & \int_{O_t\times \cR}  \rho_\inv^{\frac12} W_t^*(\omega)V_{\omega_{t+1}}^*V_{\omega_{t+1}}W_t(\omega)\rho_\inv^{1/2}\, \d\mu^{\otimes t+1} (\omega).
\end{eqnarray*}
Then the consistency condition of Equation~\eqref{eq:tracepreserving} implies
\begin{eqnarray*}
  \mathbb E_{\rho_\inv}[M_{t+1}\one_{O_t}] & = & \int_{O_t}  \rho_\inv^{\frac12} W_t^*(\omega)W_t(\omega)\rho_\inv^{1/2}\, \d\mu^{\otimes t} (\omega)\\
 & = & \int_{O_t} M_t \d\P^{\rho_\inv}\\
 &=& \mathbb E_{\rho_\inv}[M_{t}\one_{O_t}].
 \end{eqnarray*}
This proves $(M_t)_t$ is a $\P^{\rho_\inv}$ martingale.

The convergence then follows from the general theory of Banach space valued martingales, see e.g. \cite{Pi16}.  Since the space  of trace class operators is a separable dual (it is separable since $\cH_S$ is separable and it is the dual of the space of compact operators on $\cH_S$) it is sufficient to prove that $(M_t)_t$ is bounded in $L^1$ to get almost sure convergence in trace class norm, and that it is uniformly integrable to get the $L^1$ convergence, see Theorem 2.9 and Corollary 2.15 in \cite{Pi16}. These properties follow directly from the definition of $M_t$ because $\|M_t(\omega)\|_1=\tr(M_t(\omega))=1$ for any $\omega\in \Omega$.
\end{proof}

\begin{remark} As mentioned, the use of $M_t$ is inspired by the finite dimensional case. The main difference is the presence of the two $\rho_\inv^{1/2}$ factors in the numerator. The reason for adding these extra terms is related to the theory of Banach space valued martingales. While we can apply it in the space $\cJ(\cH_\cS)$ of trace class operators, it is not valid in the space of bounded operators $\cB(\cH_\cS)$. The introduction of these factors will require some care later on: we will obtain Equation~\eqref{eq:purification} first for a suitable dense subset of initial states and then proceed via an approximation argument.
\end{remark}

\begin{corollary}\label{cor:Mn-convergence} Let $\rho\in\cJ_1$. Then $(M_t)_t$ converges in trace norm, $\mathbb P^{\rho}$-almost surely and in $L^1$, towards $M_\infty$.
\end{corollary}

\begin{proof} Since $\rho_\inv$ is faithful, for all $\rho\in\mathcal J_1$, Proposition~\ref{prop:absolutecontinuity} implies $\mathbb P^\rho\ll\mathbb P^{\rho_\inv}$. Hence the almost sure convergence under $\mathbb P^{\rho_\inv}$ implies the one with respect to $\mathbb P^\rho$. The $L^1$ convergence follows using Lebesgue's dominated convergence theorem and $\|M_t\|_1=\|M_\infty\|_1=1$.
\end{proof}

\begin{remark}
  Following Lemma~\ref{lem:POVM}, actually $M_\infty(\omega)\d\pp^{\rho_{\rm inv}}(\omega)=\rho_{\rm inv}^{\frac12}P(\d \omega)\rho_{\rm inv}^{\frac12}$, where we recall that $P$ is the POVM such that $\d\pp^\rho(\omega)=\tr(\rho P(\d\omega))$ for any $\rho\in \cJ_1$. 
\end{remark}

\begin{remark}\label{rem:vanishingofM_t} We can make the following observation relating $\bra n, M_t(\omega) n\ket$, hence $\P(n|\omega_1,\ldots,\omega_t)$, to the extinction time $\tau^n(\omega)$. Using again that $|n\ket$ is an eigenvector of $\rho_\inv$, it is easy to check that $\bra n, M_t(\omega) n\ket=\P(n|\omega_1,\ldots,\omega_t)=0$ if and only if $t\geq  \tau^n(\omega)$. %, i.e. $\P(n|\omega_1,\ldots,\omega_t)\neq 0$ if and only if $t<\tau^n(\omega)$. 
As a consequence, 
\begin{equation}\label{eq:vanishingMinfini}
\bra n, M_\infty(\omega) n\ket=0, \qquad \forall \omega \in \{\tau^n<+\infty\}.
\end{equation}
\end{remark}

Our main result about $(M_t)_t$, which is the key step to prove Theorem~\ref{thm:purification}, is the following.
\begin{proposition}\label{prop:Mn-martingale} If the non-resonant condition holds and $\beta>0$ there exists a random variable $n_\infty$ valued in $\mathbb N$ such that $$M_\infty=|n_\infty\ket\bra n_\infty|.$$ Moreover for any $\rho\in \cJ_1$,
  $$\d\pp^{\rho}=\tfrac{\bra n_\infty,\rho\,n_\infty\ket}{\bra n_\infty,\rho_\inv\,n_\infty\ket}\d\pp^{\rho_\inv}$$
  and, for any $n\in \nn$, $\pp^\rho(n_\infty=n)=\bra n, \rho\ n \ket$.
\end{proposition}

\begin{remark}
  This proposition is the analogue of Proposition 2.2 in \cite{BFPP19}.
\end{remark}

\begin{remark} Coming back to our probabilistic interpretation of $M_t$, the quantity $\bra n, M_\infty(\omega) n\ket$ is the conditional probability $\P(n|\omega)$ to start in the Fock state $|n\ket\bra n|$ if the initial state is distributed according to the invariant distribution $\nu_\inv$ and we have observed the full sequence $\omega$ of measurement outcomes. The result of Proposition~\ref{prop:Mn-martingale} says that $\P(n|\omega)$ is nothing but the delta measure at $n_\infty(\omega)$: the knowledge of the full sequence of outcomes tells us where we started from (if the initial state is distributed according to $\nu_\inv$).
\end{remark}

\noindent The rest of this section is devoted to the proof of Proposition~\ref{prop:Mn-martingale}.

\smallskip

In the finite dimensional situation, the fact that $M_\infty$ is a rank one projection holds if and only if a so-called purification condition holds. In our case the non-resonant condition implies the following lemma which plays the role of purification condition, see the remark below. 
\begin{lemma}\label{lem:purificationassumption}
  Suppose the non-resonant condition holds. Then, for any $n\neq m$ there exists $s\in \nn$ and $w\in \mathcal R^s$ such that
  $$\mathbb P^{|m\ket\bra m|}(\Lambda_w)\neq \mathbb P^{|n\ket\bra n|}(\Lambda_w),$$
  where we recall that $\Lambda_w\in \cO_s$ is the cylinder set associated to $w\in\cR^s$.
\end{lemma}
\begin{proof}
 Assume without loss of generality that $m>n$, and set $w=(-,+)^m$. On the one hand, the non resonant condition implies $\pp^{|m\ket\bra m|}(\Lambda_w)=\prod_{s=0}^{m-1}p(m-s,m-s-1)>0$. On the other hand $\pp^{|n\ket\bra n|}(\Lambda_w)=\prod_{s=0}^{m-1}p(n-s,n-s-1)$ and $m>n$ implies $p(n-m,n-m-1)=0$ so that $\pp^{|n\ket\bra n|}(w)=0$ and the lemma is proved.
\end{proof}

\begin{remark}\label{rem:purificationcondition} In finite dimension the purification assumption is, see \cite{BFPP19}: 
\begin{quote} {\bf (Pur)} if $\pi$ is an orthogonal projection such that, for any $\omega\in\Omega$ and $t\geq 1$, $\ds \pi W_t^*(\omega)W_t(\omega)\pi =\lambda \pi$ for some $\lambda\in \R_+$ then $\pi$ is of rank one. 
\end{quote}
If $\pi$ is an orthogonal projection, the condition ``$\pi W_t^*(\omega)W_t(\omega)\pi =\lambda \pi$ for some $\lambda\in \R_+$'' is equivalent to ``there exists $\lambda\in\R_+$ such that $\|W_t(\omega)\phi\|^2=\lambda$  for any unit vector $\phi\in\Ran(\pi)$''. That this holds for all $t$ and $\omega$ therefore means that any initial state in the range of $\pi$ leads to the same distribution of probability on the sequence of outcomes, i.e. if $\phi,\psi \in\Ran(\pi)$ then $\P^{|\phi\ket\bra\phi|}=\P^{|\psi\ket\bra\psi|}$. 

The above lemma says that any two different Fock states lead to different probability distributions on the sequence of outcomes. In our case that is sufficient to ensure purification as expressed in Theorem~\ref{thm:purification}.
\end{remark}

In the sequel we shall use computation techniques which are reminiscent of the usual Bayes' rule and total probability formula. Although it is not necessary, we believe that it will make the computation more transparent if write it in a classical probability theory language. We have already introduced
\[
\P(n;\omega_1,\ldots,\omega_t)=\P\left(|n\ket\bra n|, \Lambda_{(\omega_1,\ldots,\omega_t)}\right),\ \P(\omega_1,\ldots,\omega_t)=\P\left(\cJ_1,\Lambda_{(\omega_1,\ldots,\omega_t)}\right)
 \ \mbox{and} \ \P(n|\omega_1,\ldots,\omega_t),
\]
with the obvious interpretation, see Equations~(\ref{eq:probastateoutcomeinvariantmeasure})-(\ref{eq:Mtinterpretation}). We further denote, for any $n\in\N$, $s,t\in\mathbb N$ and $(\omega_1,\ldots,\omega_{t+s})\in\cR^{t+s}$,  
\begin{equation}\label{def:probanotation1}
\P(\omega_{t+1},\ldots,\omega_{t+s}|\omega_1,\ldots,\omega_t)\equiv \frac{\P(\omega_1,\ldots,\omega_{t+s})}{\P(\omega_1,\ldots,\omega_t)} = \frac{\tr(W_{t+s}(\omega)^*W_{t+s}(\omega)\rho_\inv)}{\tr(W_t^*(\omega)W_t(\omega)\rho_\inv)} , 
\end{equation}
\begin{equation}\label{def:probanotation2}
\P(n;\omega_{t+1},\ldots,\omega_{t+s}|\omega_1,\ldots,\omega_t)\equiv \frac{\P(n;\omega_1,\ldots,\omega_{t+s})}{\P(\omega_1,\ldots,\omega_t)} = \frac{\bra n,\rho_\inv^{1/2} W_{t+s}(\omega)^*W_{t+s}(\omega)\rho_\inv^{1/2}n\ket}{\tr(W_t^*(\omega)W_t(\omega)\rho_\inv)} , 
\end{equation}
\begin{equation}\label{def:probanotation3}
\P(\omega_1,\ldots,\omega_t | n) \equiv \frac{\P(n;\omega_1,\ldots,\omega_t)}{\P(n)} = \frac{\bra n,\rho_\inv^{1/2} W_t(\omega)^*W_t(\omega)\rho_\inv^{1/2}n\ket}{\bra n, \rho_\inv n\ket} = \bra n,W_t^*(\omega) W_t(\omega) n\ket,
\end{equation}
and
\begin{equation}\label{def:probanotation4}
\P(\omega_{t+1},\ldots,\omega_{t+s}|n;\omega_1,\ldots,\omega_t)\equiv \frac{\P(n;\omega_1,\ldots,\omega_{t+s})}{\P(n;\omega_1,\ldots,\omega_t)} = \frac{\bra n,\rho_\inv^{1/2} W_{t+s}(\omega)^*W_{t+s}(\omega)\rho_\inv^{1/2}n\ket}{\bra n,\rho_\inv^{1/2} W_{t}(\omega)^*W_{t}(\omega)\rho_\inv^{1/2}n\ket}, 
\end{equation}
where we have used that $|n\ket$ is an eigenstate of $\rho_\inv$ in Equation~(\ref{def:probanotation3}), and where Equation~(\ref{def:probanotation4}) makes sense only if $\P(n;\omega_1,\ldots,\omega_t)\neq 0$. Note that all the other terms are always well defined because $\rho_\inv$ is faithful. By convention we have fixed $\pp(\omega_1,\dotsc,\omega_t)=1$ for $t=0$.

All these quantities have transparent probabilistic interpretations, remembering that the initial state is distributed according to the invariant distribution $\nu_\inv$. For example the quantity $\P(\omega_1,\ldots,\omega_t | n)$ is the probability to observe the first $t$ outcomes $\omega_1,\ldots,\omega_t$ if the initial state is $|n\ket$ while $\P(\omega_{t+1},\ldots,\omega_{t+s}|n;\omega_1,\ldots,\omega_t)$ is the probability to further obtain the $s$ outcomes $\omega_{t+1},\ldots,\omega_{t+s}$ if we have previously observed $\omega_1,\ldots,\omega_t$ and started from state $|n\ket$. Note that if the latter event occurs then at time $t$ the system is in the state $|N_t(n,\omega)\ket$, see Equation~(\ref{eq:evolvedfockstate}), so that one expects that
\begin{equation}\label{eq:probastartingtimet}
\P(\omega_{t+1},\ldots,\omega_{t+s}|n;\omega_1,\ldots,\omega_t) = \P\left(\omega_{t+1},\ldots,\omega_{t+s} \big| N_t(n,\omega)\right).
\end{equation}
The following lemma shows that this is indeed the case and that the above notations are consistent in the sense that they indeed allow to use the classical Bayes' rule and total probability law.
\begin{lemma} Let $s,t\in\mathbb N^*$ and $(\omega_1,\ldots,\omega_{t+s})\in\cR^{t+s}$.
\begin{enumerate}
\item If $n\in\N$ is such that $\P(n|\omega_1,\ldots,\omega_t)\neq 0$, i.e. $t<\tau^n(\omega)$, then Equality~(\ref{eq:probastartingtimet}) holds and we have the following Bayes' rule
\begin{equation}\label{eq:bayes}
\frac{\P(n|\omega_1,\ldots,\omega_{t+s})\times \P(\omega_{t+1},\ldots,\omega_{t+s}|\omega_1,\ldots,\omega_t)}{\P(n|\omega_1,\ldots,\omega_t)}  = \P\left(\omega_{t+1},\ldots,\omega_{t+s} \big| N_t(n,\omega)\right).
\end{equation}
\item The total probability formula 
\begin{equation}\label{eq:totalproba}
\P(\omega_{t+1},\ldots,\omega_{t+s}|\omega_1,\ldots,\omega_t) = \sum_k  \P\left(\omega_{t+1},\ldots,\omega_{t+s} \big| N_t(k,\omega)\right) \P(k|\omega_1,\ldots,\omega_t)
\end{equation}
holds. In the right-hand side, $N_t(k,\omega)=-1$ when $\P(k|\omega_1,\ldots,\omega_t)= 0$, i.e. $t\geq \tau^k(\omega)$. We therefore set $\pp(\omega_{t+1},\dotsc,\omega_{t+s}|-1)=0$ without loss of generality.
\end{enumerate}
\end{lemma}

\begin{proof} 1) Using that $|n\ket$ is an eigenstate of $\rho_\inv$ we immediately have from Equation~(\ref{def:probanotation4})
\begin{eqnarray*}
\P(\omega_{t+1},\ldots,\omega_{t+s}|n;\omega_1,\ldots,\omega_t) & = & \frac{\bra n, W_{t+s}(\omega)^*W_{t+s}(\omega) n\ket}{\bra n, W_{t}(\omega)^*W_{t}(\omega) n\ket} \\
 & = & \frac{\bra n, W_t(\omega)^*(V_{\omega_{t+s}}\cdots V_{\omega_{t+1}})^* V_{\omega_{t+s}}\cdots V_{\omega_{t+1}} W_t(\omega) n\ket}{\|W_{t}(\omega) n\ket\|^2} \\
 & = & \bra N_t(n,\omega), (V_{\omega_{t+s}}\cdots V_{\omega_{t+1}})^* V_{\omega_{t+s}}\cdots V_{\omega_{t+1}} N_t(n,\omega)\ket \\
 & = & \P\left(\omega_{t+1},\ldots,\omega_{t+s} \big| N_t(n,\omega)\right),
\end{eqnarray*}
where we have used Remark~\ref{rk:W_t, N_t} in the third line.
Bayes' rule~(\ref{eq:bayes}) is then a direct computation left to the reader using Equations~(\ref{eq:Mtinterpretation}),~(\ref{def:probanotation1}) and~(\ref{def:probanotation4})

2) We can write
\begin{eqnarray*}
\P(\omega_{t+1},\ldots,\omega_{t+s}|\omega_1,\ldots,\omega_t) & = & \frac{\tr(W_{t+s}(\omega)^*W_{t+s}(\omega)\rho_\inv)}{\tr(W_t^*(\omega)W_t(\omega)\rho_\inv)} \\
 & = & \sum_k \frac{\bra k,\rho_\inv^{1/2} W_{t+s}(\omega)^*W_{t+s}(\omega)\rho_\inv^{1/2}k\ket}{\tr(W_t^*(\omega)W_t(\omega)\rho_\inv)} \\
 & = & \sum_k \frac{\P(k;\omega_1,\ldots,\omega_{t+s})}{\P(\omega_1,\ldots,\omega_t)} \\
 & = & \sum_k \frac{\P\left(\omega_{t+1},\ldots,\omega_{t+s} \big| N_t(k,\omega)\right) \P(k;\omega_1,\ldots,\omega_t)}{\P(\omega_1,\ldots,\omega_t)} \\
 & = & \sum_k  \P\left(\omega_{t+1},\ldots,\omega_{t+s} \big| N_t(k,\omega)\right) \P(k|\omega_1,\ldots,\omega_t),
\end{eqnarray*}
where we have used the fact that $(|k\ket)_k$ is an orthonormal basis to obtain the 2nd line, Equation~(\ref{def:probanotation2}) for the 3rd one, Equations~(\ref{def:probanotation4})-(\ref{eq:probastartingtimet}) for the 4th one and Equation~(\ref{eq:Mtinterpretation}) for the last one.
\end{proof}

\noindent \emph{Proof of Proposition~\ref{prop:Mn-martingale}.} Let $n\in\N$. The sequence $(\langle n,M_t n\rangle)_t$ is then a bounded real martingale. Reproducing some arguments of the proof of Proposition 2.2 of \cite{BFPP19}, in particular Equation~(17) and the arguments leading to the two following displayed equations, we have that for all $s\in\mathbb N$
\begin{equation}\label{eq:conditionalexpectationlimit}
\lim_{t\to\infty}\mathbb E^{\rho_\inv}[\vert \langle n,M_t n\rangle-\langle n,M_{t+s} n\rangle\vert \, \vert \, \mathcal O_t]=0, \quad \pp^{\rho_{\inv}}\mbox{-a.s.}
\end{equation}

Let us write this conditional expectation using Equations~(\ref{eq:probastateoutcomeinvariantmeasure})-(\ref{eq:Mtinterpretation}) and~(\ref{def:probanotation1})-(\ref{def:probanotation4}). For any $s,t\in\N^*$ and $\omega\in\Omega$ we have
\[
\vert \langle n,M_t (\omega) n\rangle- \langle n,M_{t+s} (\omega) n\rangle\vert=\left| \P(n | \omega_1,\ldots,\omega_t) - \P(n | \omega_1,\ldots,\omega_{t+s})\right|.
\]
Taking the conditional expectation, for any $s,t\in\N^*$ and $\pp^{\rho_\inv}$-almost all $\omega_1,\ldots,\omega_t$, we get
\begin{eqnarray}\label{eq:conditionalexpectation}
 \lefteqn{ \E^{\rho_\inv} \left[\left|\langle n,M_t n\rangle-\langle n,M_{t+s} n\rangle\right| | \mathcal O_t\right] (\omega_1,\ldots,\omega_t)}\nonumber\\ 
 & = &   \int_{\cR^s} \Big| \P(n | \omega_1,\ldots,\omega_t) - \P(n | \omega_1,\ldots,\omega_t,\zeta_1,\ldots,\zeta_{s}) \Big| \times \P(\zeta_1,\ldots,\zeta_s| \omega_1,\ldots,\omega_t)\, \d\mu^{\otimes s}(\zeta_1,\ldots,\zeta_s) \nonumber \\
 & = &  \int_{\cR^s} \P(n | \omega_1,\ldots,\omega_t) \Big| \P(\zeta_1,\ldots,\zeta_s| \omega_1,\ldots,\omega_t) - \P(\zeta_1,\ldots,\zeta_s|N_t(n,\omega)) \Big| \, \d\mu^{\otimes s}(\zeta_1,\ldots,\zeta_s) \nonumber \\
 & = &  \int_{\cR^s} \P(n | \omega_1,\ldots,\omega_t) \Big| \sum_k \P(\zeta_1,\ldots,\zeta_s | N_t(k,\omega))\P(k|\omega_1,\ldots,\omega_t) - \P(\zeta_1,\ldots,\zeta_s|N_t(n,\omega)) \Big| \nonumber \\
 &  & \hspace*{130mm} \d\mu^{\otimes s}(\zeta_1,\ldots,\zeta_l) \nonumber \\
 & = &  \int_{\cR^s} \bra n,M_t(\omega)n\ket \Big| \sum_k \|V_{\zeta_s}\cdots V_{\zeta_1} |N_t(k,\omega)\ket\|^2 \bra k,M_t(\omega)k\ket - \|V_{\zeta_s}\cdots V_{\zeta_1} |N_t(n,\omega)\ket\|^2 \Big|  \\
 & &  \hspace*{130mm} \d\mu^{\otimes s}(\zeta_1,\ldots,\zeta_l),\nonumber
\end{eqnarray}
where we have successively used Equation~(\ref{def:probanotation1}), Bayes' rule~(\ref{eq:bayes}), the total probability formula~(\ref{eq:totalproba}) and finally Equations~(\ref{eq:Mtinterpretation}) and~(\ref{def:probanotation3}) in the last step.

Using Equation~(\ref{eq:conditionalexpectationlimit}), since $\mu$ has finite support and all the terms in Equation~(\ref{eq:conditionalexpectation}) are non-negative, for $\P^{\rho_\inv}$-almost all $\omega$ we have
\begin{equation}\label{eq:pointwiselimit}
\lim_{t\to \infty} \bra n,M_t(\omega)n\ket \Big| \sum_k \|V_{\zeta_s}\cdots V_{\zeta_1} | N_t(k,\omega)\ket\|^2 \bra k,M_t(\omega)k\ket - \|V_{\zeta_s}\cdots V_{\zeta_1} |N_t(n,\omega)\ket\|^2 \Big| =0
\end{equation}
for all $s\in\N^*$ and $\zeta_1,\ldots,\zeta_s\in \cR$.

Now, on the one hand, Corollary~\ref{cor:Mn-convergence} implies for any $k$ and $\P^{\rho_\inv}$-almost every $\omega$, 
$$\lim_{t\to\infty} \bra k,M_t(\omega)k\ket = \bra k,M_\infty(\omega)k\ket.$$
On the other hand, it follows from Proposition~\ref{prop:recurrenceqtraj} that for $\P^{\rho_\inv}$-almost every $\omega$ there exists a subsequence $(t_j)_j$ such that for all $k$ either $\lim_j N_{t_j}(k,\omega)=k$, if $\tau^k(\omega)=+\infty$, or $\lim_j N_{t_j}(k,\omega)=-1$, if $\tau^k(\omega)<+\infty$. Since $\bra k,M_\infty(\omega)k\ket=0$ whenever $\tau^k(\omega)<+\infty$, see Equation~(\ref{eq:vanishingMinfini}), for $\pp^{\rho_\inv}$-almost all $\omega$ there exists a subsequence $(t_j)_j$ such that for any $s\in \nn^*$ and $\zeta_1,\dotsc,\zeta_s\in \cR$,
\[
\lim_{j\to+\infty}  \|V_{\zeta_s}\cdots V_{\zeta_1} | N_{t_j}(k,\omega)\ket\|^2 \bra k,M_{t_j}(\omega)k\ket =  \|V_{\zeta_s}\cdots V_{\zeta_1} | k\ket\|^2 \bra k,M_\infty(\omega)k\ket, \quad \forall k
\]
where  we have set $\|V_{\zeta_s}\cdots V_{\zeta_1} |-1\ket\|=0$.

Considering the limit in Equation~(\ref{eq:pointwiselimit}) along this subsequence we hence obtain that, $\P^{\rho_\inv}$-almost surely, for all $s\in\N^*$ and $\zeta_1,\ldots,\zeta_s\in \cR$
\[
\bra n,M_\infty(\omega)n \ket  \Big| \sum_k \|V_{\zeta_s}\cdots V_{\zeta_1} \vert k\ket\|^2 \bra k,M_\infty(\omega)k\ket - \|V_{\zeta_s}\cdots V_{\zeta_1} \vert n\ket\|^2 \Big| =0
\]
where we have used Proposition~\ref{prop:recurrenceqtraj} for both the $N_{t_j}(k,\omega)$ and $N_{t_j}(n,\omega)$. Note that this expression is also valid if we had $\lim_j N_{t_j}(k,\omega)=-1$ for some $k\in \nn$. 

At this stage all the computations and arguments are independent of the choice of $n$ which has been fixed at the beginning of the proof. Thus, for any $n,m\in\N$,
\begin{eqnarray*}
\bra n,M_\infty(\omega)n \ket  \Big( \sum_k \|V_{\zeta_s}\cdots V_{\zeta_1} \vert k\ket\|^2 \bra k,M_\infty(\omega)k\ket - \|V_{\zeta_s}\cdots V_{\zeta_1} \vert n\ket\|^2 \Big) & = & 0,\\
\bra m,M_\infty(\omega)m \ket  \Big( \sum_k \|V_{\zeta_s}\cdots V_{\zeta_1} \vert k\ket\|^2 \bra k,M_\infty(\omega)k\ket - \|V_{\zeta_s}\cdots V_{\zeta_1} \vert m\ket\|^2 \Big) & = & 0.
\end{eqnarray*}
Multiplying the first identity by $\bra m,M_\infty(\omega)m \ket$, the second by  $\bra n,M_\infty(\omega)n \ket$ and substracting the two this yields, for any $m,n\in\N$, $s\in\N^*$ and $\zeta_1,\ldots,\zeta_s$,
\begin{equation}\label{eq:beforeusingpurification}
\bra n,M_\infty(\omega)n \ket \times \bra m,M_\infty(\omega)m \ket \times \Big( \|V_{\zeta_s}\cdots V_{\zeta_1} \vert m\ket\|^2 - \|V_{\zeta_s}\cdots V_{\zeta_1} \vert n\ket\|^2 \Big) =0.
\end{equation}
At this stage we can invoke Lemma~\ref{lem:purificationassumption} which plays the role of the purification assumption. Indeed,
$$\|V_{\zeta_s}\cdots V_{\zeta_1}\vert m\ket\|^2=\pp^{|m\ket\bra m|}(\Lambda_{\zeta_1,\dotsc,\zeta_s}).$$
Taking $(\zeta_1,\dotsc,\zeta_s)=w$ such that $\pp^{|m\ket\bra m|}(\Lambda_w)\neq \pp^{|n\ket\bra n|}(\Lambda_w)$ as allowed by Lemma~\ref{lem:purificationassumption}, Equation~\eqref{eq:beforeusingpurification} implies that, for any $n\neq m$,
\begin{equation}\label{eq:afterusingpurification}
\bra n,M_\infty(\omega)n \ket \times \bra m,M_\infty(\omega)m \ket =0, \quad \P^{\rho_\inv}\as
\end{equation}
Therefore, for $\pp^{\rho_\inv}$-almost all $\omega$, there exists at most one $n\in \nn$ such that $\bra n,M_\infty(\omega)n \ket\neq 0$. Since $M_t\in \cJ_1$ for all $t\in \nn$ and $(M_t)_t$ converges in trace norm we also have, $\pp^{\rho_\inv}$-almost surely, that $M_\infty\in \cJ_1$. Hence there exists a unique $n$ such that $\bra n,M_\infty(\omega)n \ket\neq 0$ and actually $\bra n,M_\infty(\omega)n \ket=1$. We denote it $n_\infty(\omega)$. Then Cauchy-Schwartz inequality for positive semi-definite operators implies $M_\infty=|n_\infty\ket\bra n_\infty|$.

It remains to prove that $\pp^\rho(n_\infty=k)=\bra k,\rho\, k\ket$ for any $k\in \nn$ and $\rho\in \cJ_1$, and that 
$$\d\pp^{\rho}=\tfrac{\bra n_\infty,\rho\  n_\infty\ket}{\bra n_\infty,\rho_\inv n_\infty\ket}\d\pp^{\rho_\inv}.$$
We first prove $\P^{\rho_\inv}(n_\infty=k)=\bra k,\rho_\inv k\ket$. For any $k$, the quantity $\bra k, n_\infty\ket^2$ defines a Bernoulli random variable with parameter $\P^{\rho_\inv}(n_\infty=k)$. Hence 
\[
\P^{\rho_\inv}(n_\infty=k) = \E^{\rho_\inv}\left(\bra k, n_\infty\ket^2\right) = \E^{\rho_\inv}\left( \bra k,M_\infty k\ket\right).
\]
But, by definition of $M_t$, we have $\E^{\rho_\inv}(M_t) =\rho_\inv$ for all $t$, and because $(M_t)_t$ converges in $L^1$ to $M_\infty$ we also have $\E^{\rho_\inv}(M_\infty) =\rho_\inv$ and hence $\P^{\rho_\inv}(n_\infty=k)=\bra k,\rho_\inv k\ket$.

Now, for any $n\in \nn$, because $|n \ket$ is an eigenvector of $\rho_\inv$ it follows from Definitions~\ref{def:Prho} and \ref{def:Mn} that $\d\pp^{|n\ket\bra n|}|_{\cO_t}=\frac{\bra n,M_t n\ket}{\bra n, \rho_\inv n\ket}\d\pp^{\rho_\inv}|_{\cO_t}$. Since $(M_t)_t$ converges in $L^1(\pp^{\rho_\inv})$ and in trace norm to $M_\infty$, 
\begin{equation}\label{eq:change measure proof}
\d\pp^{|n\ket\bra n|}=\frac{\bra n,M_\infty n\ket}{\bra n, \rho_\inv n\ket}\d\pp^{\rho_\inv}.
\end{equation} 
Using Lemma~\ref{lem:Prho decomposition} and $\bra n,M_\infty n\ket=\mathbf{1}_{n,n_\infty}$ we get the desired change of measure formula
$$\d\pp^{\rho}=\tfrac{\bra n_\infty,\rho\  n_\infty\ket}{\bra n_\infty,\rho_\inv n_\infty\ket}\d\pp^{\rho_\inv}.$$
Finally, $\pp^\rho(n_\infty=k)=\frac{\bra k,\rho\  k\ket}{\bra k,\rho_\inv k\ket}\pp^{\rho_\inv}(n_\infty=k)=\bra k,\rho\  k\ket$ which ends the proof.
\hfill\qed

\medskip
%%%%%%%%%%%%%%%%%%%%%%%%%%%%%%%%%%%%%%%%%%%%%%%%%%%%%%%%%%%%%%%%%%%%
%%%%%%%%%%%%%%%%%%%%%%%%%%%%%%%%%%%%%%%%%%%%%%%%%%%%%%%%%%%%%%%%%%%%

\subsection{Purification}\label{ssec:purificationproof}

This section is devoted to the proof of Theorem~\ref{thm:purification}. 

Let $\ds S_t = \frac{W_t^*W_t}{\tr(W_t^*W_t\rho_\inv)}$ and $U_t$ denote the partial isometry in the polar decomposition of $W_t$ so that
\[
W_t= \sqrt{\tr(W_t^*W_t \rho_\inv)}\, U_tS_t^{1/2}.
\]
Recall that, $\P^{\rho_\inv}$-almost surely, $\tr(W_t^*W_t \rho_\inv)\neq 0$ so that $S_t$ is indeed well defined.

We remark that $S_t$ is directly related to $M_t$, namely $M_t=\rho_\inv^{1/2}S_t\rho_\inv^{1/2}$. Actually, the same argument as in Lemma~\ref{lem:Mndiagonal} shows that $S_t$ is a bounded function of the number operator $N$. As a consequence we actually have, $\P^{\rho_\inv}$-almost surely and for all $t$,  
\begin{equation}\label{eq:commutation} 
M_t^{1/2}=S_t^{1/2}\rho_\inv^{1/2}.
\end{equation}
With a slight abuse of notation, we also denote by $S_t$ the associated function of the number operator. In particular $S_t^{1/2}\vert n\rangle\langle n\vert S_t^{1/2}=S_t(n)\vert n\rangle\langle n\vert$ for some scalar $S_t(n)$ so for any Fock state one has $W_t|n\ket\bra n| W_t^*= \tr(W_t^*W_t \rho_\inv)S_t(n) U_t|n\ket \bra n| U_t^*$. Hence, if $W_t|n\ket \neq 0$ we have
\begin{equation}\label{eq:WtonFock}
\frac{W_t|n\ket\bra n| W_t^*}{\tr(W_t|n\ket\bra n| W_t^*)} = U_t|n\ket\bra n| U_t^*.
\end{equation}

\begin{lemma}\label{lem:asymptoticfock} The random variable $n_\infty$ defined in Proposition~\ref{prop:Mn-martingale} satisfies $\tau^{n_\infty}=+\infty$, $\P^{\rho_\inv}$-almost surely. As a consequence, for all $t\in\mathbb N^*$ one has 
\begin{equation}\label{eq:asymptoticfock}
U_t |n_\infty\ket\bra n_\infty| U_t^* = \left| N_t(n_\infty,.)\ket \bra N_t(n_\infty,.) \right|, 
\end{equation}
$\P^{\rho_\inv}$-almost surely.
\end{lemma}

\begin{proof} Consider the quantum trajectory associated to the initial state $\rho_\inv$. Then for all $t$ one has
\[
\rho_t = \frac{W_t\rho_\inv W_t^*}{\tr(W_t^*W_t\rho_\inv)} = U_tS_t^{1/2} \rho_\inv S_t^{1/2} U_t^* = U_tM_tU_t^*.
\]
Since $U_t$ is a partial isometry, it follows from Lemma~\ref{lem:Mn-martingale} and Proposition~\ref{prop:Mn-martingale} that
\[
\lim_{t\to+\infty} \|\rho_t - U_t|n_\infty\ket\bra n_\infty| U_t^*\|_1=0, \quad \P^{\rho_\inv}-a.s.
\]
Because $\rho_t$ is a state one has $\|U_t|n_\infty\ket\|=1$, i.e. $W_t|n_\infty\ket\neq 0$, for all $t$ (this holds apriori only for $t$ large enough but hence for all $t$ by definition of the $W_t$'s). The lemma then follows from Equations~(\ref{eq:extinctiontime})-(\ref{eq:evolvedfockstate})  and~(\ref{eq:WtonFock}).
\end{proof}

We will prove that Theorem~\ref{thm:purification} holds with $n_\infty$ given by Proposition~\ref{prop:Mn-martingale}. Its law was given in the proposition so together with the last lemma it proves the first part of the theorem. In view of Equation~(\ref{eq:asymptoticfock}) it only remains to prove that for any initial state $\rho\in\cJ_1$ one has 
\[
\lim_{t\to\infty} \|\rho_t - U_t|n_\infty\ket\bra n_\infty| U_t^*\|_1=0, \quad \P^\rho-a.s.
\]
Using the polar decomposition of $W_t$ we may write $\ds \rho_t=U_t \frac{S_t^{1/2}\rho S_t^{1/2}}{\tr(S_t\rho)} U_t^*$ and because $U_t$ is a partial isometry it actually suffices to prove that
\begin{equation}\label{eq:purificationequiv}
\lim_{t\to\infty} \left\|\frac{S_t^{1/2}\rho S_t^{1/2}}{\tr(S_t\rho)} - |n_\infty\ket\bra n_\infty| \right\|_1=0, \quad \P^\rho-a.s.
\end{equation}
Formally, using Equation~(\ref{eq:commutation}), one would like to write $S_t^{1/2}\rho S_t^{1/2}$ as $M_t^{1/2}\rho_\inv^{-1/2}\rho\rho_\inv^{-1/2}M_t^{1/2}$ and then use the result of Section~\ref{sec:martigale} about $(M_t)_t$. However $\rho_\inv^{-1/2}\rho\rho_\inv^{-1/2}$ may not be well defined as a bounded operator because $\rho_\inv^{-1/2}$ is not. We thus first consider states of the form $\rho=\rho_\inv^{1/2}A\rho_\inv^{1/2}$ where $A\in \cB(\cH_\cS)$. Because the set $\{\rho_\inv^{1/2}A\rho_\inv^{1/2},\, A\in\cB(\cH_\cS)\mbox{ s.t. }A\geq 0 \mbox{ and }\tr[\rho_\inv A]=1\}$ is dense in $\cJ_1$ we can then use an approximation argument. 

\begin{remark} Finite rank operators are dense in $\cJ_1$ and because the Fock vectors $|n\ket$ form an orthonormal basis so is ${\rm Span}\{|n\ket\bra m|, \ n,m\in\N\}$. This subspace is contained in $\{\rho_\inv^{1/2}A\rho_\inv^{1/2},\, A\in\cB(\cH_\cS)\}$ so that the latter is indeed dense in $\cJ_1$.
\end{remark}

\begin{lemma}\label{lem:purificationdensesubset} Let $\rho\in\cJ_1$ and suppose there exists $A\in \cB(\cH_\cS)$ such that $\rho=\rho_\inv^{1/2} A \rho_\inv^{1/2}$. Then 
\[
\lim_{t\to\infty} S_t^{1/2}\rho S_t^{1/2} = \frac{\bra n_\infty, \rho \, n_\infty\ket}{\bra n_\infty,\rho_\inv n_\infty\ket} \times |n_\infty\ket\bra n_\infty|  , \quad \P^{\rho_\inv}-a.s,
\]
where the convergence is in trace norm. In particular $$\ds \lim_{t\to \infty} \tr(S_t\rho) = \frac{\bra n_\infty, \rho \, n_\infty\ket}{\bra n_\infty,\rho_\inv n_\infty\ket},\quad \P^{\rho_\inv}-a.s.$$
so that Equation~\eqref{eq:purificationequiv} holds.
\end{lemma}

\begin{proof} Since $\rho=\rho_\inv^{1/2} A\rho_\inv^{1/2}$, it follows from Equation~(\ref{eq:commutation}) that
\[
S_t^{1/2}\rho S_t^{1/2} = M_t^{1/2} A M_t^{1/2}.
\]
Using Lemma~\ref{lem:Mn-martingale} and Proposition~\ref{prop:Mn-martingale} we have $\ds \lim_{t\to\infty} M_t=|n_\infty\ket\bra n_\infty|$, $\P^{\rho_\inv}$-almost surely and in trace norm. Hence by continuity of the map $\cJ_1 \ni X\mapsto |X|^{1/2} \in \cJ_2$, see e.g. \cite[Example~2 Page~28]{Si05}, $(M_t^{1/2})_t$ converges $\P^{\rho_\inv}$-almost surely to $|n_\infty\ket\bra n_\infty|$ in the Hilbert-Schmidt topology. It implies $M_t^{1/2}AM_t^{1/2}$ converges weakly to $M_\infty^{1/2} A M_\infty^{1/2}$ and $\lim_t\|M_t^{1/2}AM_t^{1/2}\|_1=\|M_\infty^{1/2} A M_\infty^{1/2}\|_1$ $\pp^{\rho_\inv}$-almost surely. Since $M_t^{1/2}AM_t^{1/2}$ and $M_\infty^{1/2}AM_\infty^{1/2}$ are positive semi-definite, \cite[Theorem~2.20]{Si05} implies, $\P^{\rho_\inv}$-almost surely and in trace norm, 
\[
\lim_{t\to+\infty} M_t^{1/2}AM_t^{1/2} = \bra n_\infty,An_\infty\ket \times |n_\infty\ket\bra n_\infty| =  \frac{\bra n_\infty, \rho \, n_\infty\ket}{\bra n_\infty,\rho_\inv n_\infty\ket}\times |n_\infty\ket\bra n_\infty|,
\]
where we have used that the $|n_\infty\ket$ are eigenstates of $\rho_\inv$ in the last step. 
\end{proof}

To extend the previous lemma to arbitrary state $\rho$ we require the following one.
\begin{lemma}\label{lem:Sn-weakcvg} The sequence $(S_t)_t$ converges $\P^{\rho_\inv}$-a.s. in $\cB(\cH_\cS)$ to $S_\infty:=\bra n_\infty, \rho_\inv n_\infty\ket^{-1} \, |n_\infty\ket\bra n_\infty|$ in the weak$-*$ topology. In particular the sequence $(S_t)_t$ is bounded in $\cB(\cH_\cS)$.
\end{lemma}
\begin{proof} Using the same argument as in Lemma~\ref{lem:Mn-martingale}, for any $\rho\in \cJ_1$, the sequence $(\tr(S_t \rho))_t$ is a non-negative $(\cO_t)_t$-martingale with respect to $\pp^{\rho_\inv}$. Hence it converges $\P^{\rho_\inv}$-almost surely. Because any trace class operator is the linear combination of at most $4$ elements of $\cJ_1$ this proves the weak$-*$ convergence and thus the boundedness of the sequence $(S_t)_t$. It remains to identify the limit.

When $\rho$ is of the form $\rho=\rho_\inv^{1/2}A\rho_\inv^{1/2}$ we know by Lemma~\ref{lem:purificationdensesubset} that
\[
\lim_{t\to+\infty} \tr(S_t \rho) = \frac{\bra n_\infty, \rho \, n_\infty\ket}{\bra n_\infty,\rho_\inv n_\infty\ket}, \quad \P^{\rho_\inv}-a.s.
\]
Since the set $\{\rho_\inv^{1/2}A\rho_\inv^{1/2},\, A\in\cB(\cH_\cS)\}$ is dense in $\cJ_1$ the lemma holds.
\end{proof}

\noindent \emph{End of the proof of Theorem~\ref{thm:purification}.} It remains to prove that Equation~(\ref{eq:purificationequiv}) holds for any initial state $\rho$. 

Let $\rho\in\cJ_1$. Given $\varepsilon>0$ let $\tilde\rho=\rho_\inv^{1/2}A\rho_\inv^{1/2}$, $A\in\cB(\cH_\cS)$, be a state such that $\|\rho-\tilde\rho\|_1<\varepsilon$. We decompose 
\begin{equation}
  \label{eq:approx S_t rho S_t}
  \frac{S_t^{1/2}\rho S_t^{1/2}}{\tr(S_t\rho)} = \frac{S_t^{1/2} (\rho-\tilde\rho)S_t^{1/2}}{\tr(S_t\rho)} + \frac{S_t^{1/2}\tilde\rho S_t^{1/2}}{\tr(S_t\rho)}.
\end{equation}
From Lemma~\ref{lem:Sn-weakcvg}, $\ds \lim_{t\to\infty} \tr(S_t\rho)= \frac{\bra n_\infty, \rho \, n_\infty\ket}{\bra n_\infty,\rho_\inv n_\infty\ket}$ which is $\P^\rho$-a.s. non zero by Proposition~\ref{prop:Mn-martingale}.
Since $(S_t)_t$ is bounded in $\cB(\cH_\cS)$ by Lemma~\ref{lem:Sn-weakcvg} this proves that the first term on the right hand side of Equation~\eqref{eq:approx S_t rho S_t} is bounded $\P^\rho$-a.s. in trace norm by $C\varepsilon$ (the constant $C$ may be random but does not depend on $t$).

Now using Lemma~\ref{lem:purificationdensesubset} we have $\P^{\rho_\inv}$-a.s., hence $\P^\rho$-a.s. by Proposition~\ref{prop:absolutecontinuity}, and in trace norm
\[
\lim_{n\to \infty} S_t^{1/2}\tilde\rho S_t^{1/2} = \frac{\bra n_\infty, \tilde\rho \, n_\infty\ket}{\bra n_\infty,\rho_\inv n_\infty\ket} \times |n_\infty\ket\bra n_\infty|.
\]
Using again that, $\P^\rho$-a.s., $\ds \lim_{t\to \infty} \tr(S_t\rho) = \frac{\bra n_\infty, \rho \, n_\infty\ket}{\bra n_\infty,\rho_\inv n_\infty\ket}\neq 0$ we get that, $\P^\rho$-a.s. and in trace norm,
\[
\lim_{t\to \infty} \frac{S_t^{1/2}\tilde\rho S_t^{1/2}}{\tr(S_t\rho)} =  \frac{\bra n_\infty, \tilde\rho \, n_\infty\ket}{\bra n_\infty,\rho\, n_\infty\ket} \times |n_\infty\ket\bra n_\infty|.
\]
All together, $\P^\rho$-a.s.,
\[
\limsup_{n\to+\infty} \left\|\frac{S_t^{1/2}\rho S_t^{1/2}}{\tr(S_t\rho)}  - |n_\infty\ket\bra n_\infty| \right\|_1 \leq C\varepsilon +\frac{|\bra n_\infty, (\tilde\rho-\rho) \, n_\infty\ket|}{\bra n_\infty,\rho\, n_\infty\ket} \leq \left(C+\frac{1}{\bra n_\infty,\rho\, n_\infty\ket}\right) \varepsilon.
\]
Since this is true for any $\varepsilon>0$ it implies Equation~(\ref{eq:purificationequiv}). 

Finally the $L^1$ convergence in Theorem~\ref{thm:purification} follows from the almost sure convergence and the dominated convergence theorem.\hfill\qed

%%%%%%%%%%%%%%%%%%%%%%%%%%%%%%%%%%%%%%%%%%%%%%%%%%%%%%%%%%%%%%%%%%%%
%%%%%%%%%%%%%%%%%%%%%%%%%%%%%%%%%%%%%%%%%%%%%%%%%%%%%%%%%%%%%%%%%%%%
%%%%%%%%%%%%%%%%%%%%%%%%%%%%%%%%%%%%%%%%%%%%%%%%%%%%%%%%%%%%%%%%%%%%
%%%%%%%%%%%%%%%%%%%%%%%%%%%%%%%%%%%%%%%%%%%%%%%%%%%%%%%%%%%%%%%%%%%%

\section{Proof of Theorem~\ref{thm:invariantmeasure}}\label{sec:invariantmeasure}
The strategy of the proof is similar to the one in the finite dimensional case. For the sake of completeness and to be self-contained we recall some ingredients of \cite{BFPP19}. Following \cite{BFPP19} we first need to introduce another family of probability measures. Recall that $\mathfrak B$ is the Borel sigma algebra on $\cJ_1$ and $\cF=\mathfrak B\otimes \mathcal O$ is the extended sigma algebra on $\mathcal J_1 \times \Omega$, see Equation~(\ref{def:extendedalgebra}).

\begin{definition} For a given probability measure $\nu$ on $\mathcal J_1$, we define the probability measure $\P_\nu$ on $(\cJ_1\times\Omega,\cF)$ by
\begin{equation*}\label{def:Prhonu}
\P_\nu(S,O_t) := \int_{S\times O_t} \tr\left(W_t(\omega)\rho W_t^*(\omega)\right)\,\d\nu(\rho) \d\mu^{\otimes t}(\omega),
\end{equation*}
for all $t\in\mathbb N^*$, for all $S\in\mathfrak B$ and all $O_t\in\mathcal O_t.$
\end{definition}
\begin{remark} The family of measures $\P_\nu$ generalizes the probability $\P$ defined in Equation~(\ref{def:probaMinterpretation}) to interpret the quantity $\bra n,M_t(\omega) n\ket$.
\end{remark}

\noindent There is a natural connection between the $\P_\nu$'s and the probability measures $\P^\rho$ we have used so far. Given a measure $\nu$ on $\cJ_1$ let 
$$
\rho_\nu:=\mathbb E_\nu[\rho]
$$ 
be the expected state. Then $\P^{\rho_\nu}$ is the second marginal of $\P_\nu$, while its first marginal is clearly $\nu$. As an immediate consequence of Proposition~\ref{dist_var_totale}, for any bounded and $\mathcal O$-measurable function $f$ and for any measures $\nu$ and $\tilde\nu$, 
\begin{equation}\label{eq:Enulipschitz}
\vert\mathbb E_{\tilde\nu}[f]-\mathbb E_{\nu}[f]\vert\leq\Vert f\Vert_\infty\|\rho_{\tilde\nu}-\rho_\nu\|_{1}.
\end{equation}

An important example is when $\nu=\nu_\inv$. A direct computation gives that $\E_{\nu_\inv}[\rho]=\rho_\inv$ so that the second marginal of $\P_{\nu_\inv}$ is $\P^{\rho_\inv}$, i.e. for any $O\in \cO$ one has
\begin{equation*}\label{eq:Pnuinvmarginal}
\P_{\nu_\inv} (\cJ_1, O) = \P^{\rho_\inv}(O).
\end{equation*}
The following lemma characterizes $\rho_\nu$ for $\Pi$-invariant measures $\nu$. Recall that $\Pi$ denotes the Markov kernel defined in Equation~(\ref{def:markovkernel}) and $\cL$ is the CPTP map describing the evolution of the cavity without measurement, see Equation~(\ref{def:rdm}).
\begin{lemma}\label{lem:frompitoL} For any probability measure $\nu$ one has $\rho_{\nu\Pi} = \cL(\rho_\nu)$. In particular if $\nu$ is an invariant probability measure then $\rho_\nu$ in an invariant state for $\cL$.
\end{lemma}

\begin{proof} It suffices to write
\begin{eqnarray*}
\rho_{\nu\Pi} & = & \int_{\cJ_1\times \cR}  \frac{V_y\rho V_y^*}{\tr(V_y\rho V_y^*)} \tr(V_y\rho V_y^*)\, \d\mu(y)\d\nu(\rho){}\\
 & = & \int_\cR V_y \rho_\nu V_y^*\, \d\mu(y) \\
 & = & \cL(\rho_\nu)
\end{eqnarray*}
and the lemma follows.
\end{proof}

Using Theorem~\ref{thm:Lmixing} and Remark~\ref{rem:noinvstate}, the preceding lemma immediately gives the following information about invariant measures.
\begin{corollary}\label{coro:invariantmeasureinfo} Suppose the non-resonant condition holds. If $\beta>0$ any invariant probability measure $\nu$ satisfies $\rho_\nu=\rho_\inv$ and if $\beta\leq 0$ there is no invariant probability measure.
\end{corollary}

\begin{remark} The absence of invariant probability measure when $\beta\leq 0$ should not come as a surprise. Besides the absence of invariant state for the map $\cL$, the restriction to Fock states is naturally associated to a classical birth and death process exactly as in Section~\ref{ssec:classicalchain}. Following the proof of Proposition~\ref{prop:classicalmarkov} one then obtains that this Markov chain is either transient, when $\beta<0$, or null recurrent, when $\beta=0$. In any case it has no invariant probability distribution.
\end{remark}

\noindent \emph{Proof of Theorem~\ref{thm:invariantmeasure}.}
Using Equations~\eqref{eq:krausonfock} and Proposition~\ref{prop:classicalmarkov}, the probability measure defined by $\nu_\inv(\{|n\ket \bra n|\})=\mu_{Gibbs}(n)$ and $\nu_\inv(\cJ_1\setminus\{|n\ket\bra n|: n\in \nn\})=0$ is a $\Pi$-invariant measure. It remains to prove it is unique and the convergence for any initial measure.

We first prove that $\nu_\inv$ is the unique invariant probability measure. Suppose $\nu$ is an invariant probability measure. Let $f:\mathcal J_1\rightarrow\mathbb R$ be a bounded uniformly continuous function (for the $\Vert.\Vert_1$ norm induced metric). We denote by $\E_\nu$ the expectation with respect to $\P_\nu$. %In the case of a $\mathfrak B$ measurable observable there is no ambiguity issue with the expectation with respect to $\nu$ since $\nu$ is the first marginal of $\P_\nu$. 
For shortness we also denote by $\rho_{\infty,t}$ the ``asymptotic quantum trajectory'', i.e.
\[
\rho_{\infty,t}:= |N_t(n_\infty,.)\ket\bra N_t(n_\infty,.)|.
\]
As mentioned after Theorem~\ref{thm:purification} it is the quantum trajectory associated to the initial state $\rho_\infty = |n_\infty\ket\bra n_\infty|$. We will use it as a consistent estimator for the trajectory $\rho_t$.

The invariance of $\nu$ implies that for all $t$
\[
\mathbb E_{\nu}[f(\rho)]  =  \mathbb E_\nu[f(\rho_t)].
\]
On the other hand, it follows from the purification Theorem~\ref{thm:purification}, the uniform continuity of $f$ and Lebesgue dominated convergence theorem that for any probability measure $\nu$
\begin{equation}\label{eq:L1purification}
\lim_{t\to +\infty} \mathbb E_\nu[f(\rho_t)] - \E_\nu[f(\rho_{\infty,t})] = 0.
\end{equation}
Moreover, because $\rho_{\infty,t}$ is $\cO$-measurable and $\P^{\rho_\nu}$ is the second marginal of $\P_\nu$, we have for all $t$
\[
\E_\nu[f(\rho_{\infty,t})] = \E^{\rho_\nu} [f(\rho_{\infty,t})]. 
\]
Since $\nu$ is invariant Lemma~\ref{lem:frompitoL} implies $\rho_\nu=\rho_\inv$ and all together,
\[
\mathbb E_{\nu}[f(\rho)] = \lim_{t\to+\infty} \E^{\rho_\inv} [f(\rho_{\infty,t})].
\]
The right-hand side does not depend on $\nu$ so that $\mathbb E_{\nu}[f(\rho)]=\mathbb E_{\nu_\inv}[f(\rho)]$ for any uniformly bounded continuous function $f$. This proves that $\nu=\nu_\inv$.

\smallskip

We now prove the convergence in Wasserstein metric of $\nu\Pi^t$ towards $\nu_\inv$. Let $f\in \operatorname{Lip}_1(\cJ_1)$. By the translation invariance of the dual definition of Wasserstein metric in Equation~\eqref{eq:def Wasserstein} we can assume $f(\rho_{\rm inv})=0$. Since $\sup_{\rho,\varrho\in \cJ_1}\|\rho-\varrho\|_1=2$ we get $\|f\|_\infty\leq 2+|f(\rho_{\rm inv})|=2$ so that $f$ is bounded uniformly continuous. For any $s,u$ such that $t=s+u$, 
\begin{eqnarray}\label{eq:convinlawfalse}
\vert \mathbb E_{\nu\Pi^t}[f(\rho)]-\mathbb E_{\nu_\inv}[f(\rho)]\vert & = & \left| \E_{\nu\Pi^s}[f(\rho_u)] - \E_{\nu_\inv}[f(\rho_u)]\right| \nonumber\\
 & \leq & \vert \E_{\nu\Pi^s}[f(\rho_u)] -\E_{\nu\Pi^s}[f(\rho_{\infty,u})] \vert+\vert \mathbb E_{\nu_\inv}[f(\rho_{\infty,u})]-\mathbb E_{\nu_\inv}[f(\rho_u)]\vert\nonumber\\
&&+\vert \mathbb E_{\nu\Pi^s}[f(\rho_{\infty,u})]-\mathbb E_{\nu_\inv}[f(\rho_{\infty,u})]\vert \nonumber\\
  & \leq & \E_{\nu\Pi^s}[\|\rho_u-\rho_{\infty,u}\|_1]+\mathbb E_{\nu_\inv}[\|\rho_{\infty,u}-\rho_u\|_1] \nonumber\\
&&+\vert \mathbb E_{\nu\Pi^s}[f(\rho_{\infty,u})]-\mathbb E_{\nu_\inv}[f(\rho_{\infty,u})]\vert,
\end{eqnarray}
where we have used that $f\in \operatorname{Lip}_1(\cJ_1)$ in the last step.

In the third term, because $\rho_{\infty,u}$ is $\cO$ measurable, combining Equation~(\ref{eq:Enulipschitz}), Lemma~\ref{lem:frompitoL} and the fact that $\rho_{\nu_\inv}=\rho_\inv$ we have
\[
\vert \mathbb E_{\nu\Pi^s}[f(\rho_{\infty,u})]-\mathbb E_{\nu_\inv}[f(\rho_{\infty,u})]\vert \leq 2\Vert \cL^s(\rho_\nu)-\rho_\inv\Vert_1.
\]
Inserting in Equation~\eqref{eq:convinlawfalse} we therefore have, for any $s$ and $u$ such that $t=s+u$,
\[
 \vert \mathbb E_{\nu\Pi^t}[f(\rho)]-\mathbb E_{\nu_\inv}[f(\rho)]\vert \leq \E_{\nu\Pi^s}[\|\rho_u-\rho_{\infty,u}\|_1]+\mathbb E_{\nu_\inv}[\|\rho_{\infty,u}-\rho_u\|_1] +2\Vert \cL^s(\rho_\nu)-\rho_\inv\Vert_1.
 \]
Note that this upper bound is uniform in $f\in \operatorname{Lip}_1$ hence, by the Kantorovich-Rubinstein duality Equation~\eqref{eq:def Wasserstein}, it is also an upper bound for $W_1(\nu\Pi^t,\nu_\inv)$.

 Fix now $\eps>0$. From Theorem~\ref{thm:Lmixing}, there exists $s_0\in \nn$ such that $2\Vert \cL^{s_0}(\rho_\nu)-\rho_\inv\Vert_1\leq \eps/3$. 
 Using the $L^1$ convergence of Equation~(\ref{eq:purification}) twice, respectively with $\nu\Pi^{s_0}$ and $\nu_\inv$, for $u$ large enough we have
 $$\E_{\nu\Pi^s}[\|\rho_u-\rho_{\infty,u}\|_1]\leq \eps/3\quad\mbox{and}\quad \mathbb E_{\nu_\inv}[\|\rho_{\infty,u}-\rho_u\|_1]\leq \eps/3.$$
All together, for all $t$ large enough,
\[
W_1(\nu\Pi^{t},\nu_\inv) < \eps,
\]
which ends the proof.
\hfill$\Box$

%%%%%%%%%%%%%%%%%%%%%%%%%%%%%%%%%%%%%%%%%%%%%%%%%%%%%%%%%%%%%%%%%%%%
%%%%%%%%%%%%%%%%%%%%%%%%%%%%%%%%%%%%%%%%%%%%%%%%%%%%%%%%%%%%%%%%%%%%
%%%%%%%%%%%%%%%%%%%%%%%%%%%%%%%%%%%%%%%%%%%%%%%%%%%%%%%%%%%%%%%%%%%%
%%%%%%%%%%%%%%%%%%%%%%%%%%%%%%%%%%%%%%%%%%%%%%%%%%%%%%%%%%%%%%%%%%%%

\section{The resonant cases}\label{sec:resonant}

%%%%%%%%%%%%%%%%%%%%%%%%%%%%%%%%%%%%%%%%%%%%%%%%%%%%%%%%%%%%%%%%%%%%
%%%%%%%%%%%%%%%%%%%%%%%%%%%%%%%%%%%%%%%%%%%%%%%%%%%%%%%%%%%%%%%%%%%%

\subsection{Rabi sectors}

As we mentioned in Section~\ref{ssec:masermodel}, see Equation~(\ref{def:rabi}) and below, depending on the arithmetic properties of the parameters $\xi$ and $\eta$ one may have zero, one or infinitely many resonances. Up to now we have only considered the case where there is no such resonance. From a dynamical point of view it is certainly the most interesting case because it then involves an infinite dimensional system. The resonant situations are however not without interest. In particular we will exhibit particular situations where purification fails.

The resonances then induce splittings of the cavity. More precisely if $R(\xi,\eta)$ denotes the set of resonances, i.e. $$R(\xi,\eta)=\{n\in\N\, |\, \exists j\in\N, \ \xi n+\eta=j^2\}$$ and $r= 1 +{\rm Card}(R(\xi,\eta))$, the Hilbert space $\cH_\cS$ has a decomposition of the form
\begin{equation*}\label{def:RabiDecomp}
\cH_S=\bigoplus_{j=1}^r\cH_\cS^{(j)},
\end{equation*}
where $\cH_\cS^{(j)}\equiv\ell^2(I_j)$ and $\{I_j\,|\,j=1,\ldots, r\}$ is the partition of $\N$ induced by the resonances, namely
$$
\begin{array}{lcl}
I_1\equiv\N& \text{if}& R(\eta,\xi)\ \text{is empty} ,\\
I_1\equiv\{0,\ldots,n_1-1\},\ I_2\equiv\{n_1,n_1+1,\ldots\} & \text{if} &R(\eta,\xi)=\{n_1\},\\
I_1\equiv\{0,\ldots,n_1-1\},\ 
I_2\equiv\{n_1,\ldots,n_{2}-1\},\ \ldots&\text{if}&R(\eta,\xi)=\{n_1,n_2,\ldots\}.
\end{array}
$$
The resonances are exactly those values of $n$ for which the transition probabilities $p(n-1,n)$ and $p(n,n-1)$ vanish, see Equations~(\ref{eq:transitionproba-})-(\ref{eq:transitionproba+}), inducing a similar splitting of the classical Markov chain considered in Section~\ref{ssec:classicalchain}. $\cH_\cS^{(j)}$ is called the $j$-th Rabi sector, and we denote by $P_j$ the corresponding orthogonal projection.

Although not obvious at this point, there is a particular subcase when infinitely many resonances occur.
\begin{definition}\label{def:degenerate} The system is called degenerate if there exists $m\in \{0\}\cup R(\xi,\eta)$ and $n\in R(\xi,\eta)$ such that $m<n$ and $m+1,n+1\in R(\xi,\eta)$. In other words it is degenerate if there exist, at least, two Rabi sectors of dimension $1$. We shall denote $N(\xi,\eta)\equiv\{n\in\{0\}\cup R(\xi,\eta),\ n+1\in R(\xi,\eta)\}$.
\end{definition}

\begin{remark}\label{rem:degenerate} Degenerate systems exist. For example $N(724,241)=\{1,2\}$ and $N(840,1)=\{1,52\}$. It is easy to prove (see \cite{BP09}) that, for given $\xi,\eta$,  $N(\xi,\eta)$ is a finite set and that is empty in the perfectly tuned case $\eta=0$, i.e. when the frequency of the cavity mode equals the atom Bohr frequency. We refer to \cite{BP09} for more details about degenerate systems. 

\end{remark}

\medskip

For $\beta\in\R$, to each Rabi sector $\cH_\cS^{(j)}$ we associate the local Gibbs state
$$
\rho^{(j)}_\inv :=\frac{\e^{-\beta\eps N}P_j}{\tr \ (\e^{-\beta\eps N}P_j)}.
$$
Note that, for $\beta\leq 0$, $\rho^{(j)}_\inv$ is a well defined state only for finite dimensional sectors. The following theorem is the analogue of Theorem~\ref{thm:Lmixing} in the resonant situations. 

\begin{theorem}\label{thm:invstate} 1. If the system has exactly one resonance and $\beta>0$, the invariant states of $\cL$ are the convex combinations of $\rho^{(1)}_\inv$ and $\rho^{(2)}_\inv$. Moreover for any state $\rho$ one has
\begin{equation*}
\lim_{t\to\infty} \cL^t(\rho) =\tr(\rho P_1)\,\rho^{(1)}_\inv +\tr(\rho P_2)\,\rho^{(2)}_\inv,
\label{RelaxSimple}
\end{equation*}
where the limit is in the trace norm.

2. If the system has infinitely many resonances and is non-degenerate, the invariant states of $\cL$ are the convex combinations of the $\rho^{(j)}_\inv$'s. Moreover for any state $\rho$ one has
\begin{equation}
\lim_{t\to\infty} \cL^t(\rho) = \sum_{j=1}^\infty \tr(\rho P_j)\,\rho^{(j)}_\inv,
\label{RelaxFull}
\end{equation}
where the limit is again in the trace norm.
\end{theorem}

\begin{remark} 1. The above theorem is essentially proven in \cite{BP09}. The only difference is that the convergence was stated in the ergodic mean and in the weak topology. In the case of infinitely many resonances the same proof as in \cite{BP09} actually suffices. In the case of a single resonance it can be proven using exactly the same approach as for the non-resonant case in \cite{Bru14}.

2. From a spectral point of view, in all but the degenerate case, $1$ is the only peripheral eigenvalue of $\cL$ and with corresponding eigenstates given by the local Gibbs states $\rho^{(j)}_\inv$ and their convex combinations. In the degenerate case there are additional peripheral eigenvalues. They are exactly the $\e^{i(\tau \eps+\pi\xi)(n-m)}$ where $n,m\in N(\xi,\eta)$, $n\neq m$. If these numbers are different from $1$ then the convergence~(\ref{RelaxFull}) holds but only in the mean ergodic sense.

3. As we shall see below, the degenerate case has also some consequence at the level of quantum trajectories. It provides a situation where purification then does not hold.
\end{remark}

 \medskip
%%%%%%%%%%%%%%%%%%%%%%%%%%%%%%%%%%%%%%%%%%%%%%%%%%%%%%%%%%%%%%%%%%%%
%%%%%%%%%%%%%%%%%%%%%%%%%%%%%%%%%%%%%%%%%%%%%%%%%%%%%%%%%%%%%%%%%%%%

\subsection{Invariant measures}\label{ssec:resonantinvariantmeasure}

In the resonant cases one then obviously loses the uniqueness of an invariant probability measure for the Markov kernel $\Pi$ defined in~(\ref{def:markovkernel}). The probability measures
\[
\nu_\inv^{(j)}:= \sum_{n=0}^\infty \bra n,\rho_\inv^{(j)} n\ket \delta_{|n\ket\bra n|}
\] 
associated to the local Gibbs states $\rho_\inv^{(j)}$, as well as any of their convex combination, are indeed all invariant. One can then expect that our convergence Theorem~\ref{thm:invariantmeasure} can be adjusted in the spirit of the above Theorem~\ref{thm:invstate}. For finite dimensional systems, this situation corresponds to the non irreducible case considered in Appendix B of \cite{BFPP19}. Mimicking the proof in Section~\ref{sec:invariantmeasure} it is not hard to see that, provided a purification result as in Proposition~\ref{prop:Mn-martingale} holds and using Theorem~\ref{thm:invstate} instead of Theorem~\ref{thm:Lmixing}, for any given probability distribution $\nu$ on $\cJ_1$ the sequence $(\nu\Pi^t)_t$ converges in the Wasserstein distance to 
\[
\nu = \sum_{j} \tr(\rho_\nu P_j)\,\nu^{(j)}_\inv,
\]
where we recall that $\rho_\nu=\E_\nu[\rho]$.

\medskip
%%%%%%%%%%%%%%%%%%%%%%%%%%%%%%%%%%%%%%%%%%%%%%%%%%%%%%%%%%%%%%%%%%%%
%%%%%%%%%%%%%%%%%%%%%%%%%%%%%%%%%%%%%%%%%%%%%%%%%%%%%%%%%%%%%%%%%%%%

\subsection{Purification}\label{sec:purificationgeneral}

Throughout  Section~\ref{sec:purification}, and up to~(\ref{eq:beforeusingpurification}), we have used the non-resonant condition only when considering the classical birth and death process in Section~\ref{ssec:classicalchain}, and in Lemma~\ref{lem:purificationassumption} which allows to go from Equation~(\ref{eq:beforeusingpurification}) to Equation~(\ref{eq:afterusingpurification}). At the level of the classical chain, the existence of resonances is related to a lack of irreducibility. Indeed, $n\in R(\xi,\eta)$ if and only if $p(n-1,n)=p(n,n-1)=0$. In any case this classical chain is however still recurrent and the only change is in Lemma~\ref{lem:translationinvariance}. The extinction time $\tau^k=\inf\{t,\Vert W_{t}|k\ket\Vert =0\}$ is now the first time where the trajectory leaves the Rabi sector in which it started, i.e if $|k\ket \in \cH_\cS^{(j)}$ is in the $j$-th Rabi sector then
\[
\tau^k(\omega) = \inf\{t\geq 1, \ k+s_t(\omega) \notin \{n_j,\ldots,n_{j+1}-1\} \}.
\]
We do not have anymore $\{\tau^k=+\infty\}\subset \{\tau^l=+\infty\}$ when $k\leq l$, but the result of Lemma~\ref{lem:translationinvariance} holds on $\{\tau^k=+\infty\}\cap \{\tau^{k+m}=+\infty\}$ and, as a consequence, Proposition~\ref{prop:recurrenceqtraj} remains true with exactly the same proof. No matter if the system is resonant or not, Equation~(\ref{eq:beforeusingpurification}) therefore holds. 

To understand whether we have a purification result the next step is thus to see if, for all $n\neq m$ there exists  $\zeta_1,\ldots,\zeta_s$ such that 
\begin{equation}\label{eq:purificationgeneral}
\|V_{\zeta_s}\cdots V_{\zeta_1}\vert m\ket\|^2\neq \|V_{\zeta_s}\cdots V_{\zeta_1}\vert n\ket\|^2.
\end{equation}
This is exactly the step where we used Lemma~\ref{lem:purificationassumption}. If it holds true then we have Equation~(\ref{eq:afterusingpurification}) and one can copy-paste the end of Section~\ref{sec:purification}. Below we analyze Equation~(\ref{eq:purificationgeneral}) and hence derive necessary and sufficient conditions for purification, see Lemma~\ref{lem:nopurification}.

Before we embark in that endeavour let us point out that Equation~\eqref{eq:purificationgeneral} is actually necessary for purification.
\begin{proposition}
  \label{prop:nopurification}
  Let $\mathcal N$ be a maximal subset of $\nn$ such that for any $n,m\in \mathcal N$, $s\in \nn^*$ and any $\zeta_1,\ldots,\zeta_s\in \mathcal R$,
  \begin{equation}\label{eq:nopur}
    \|V_{\zeta_s}\cdots V_{\zeta_1}\vert m\ket\|^2=\|V_{\zeta_s}\cdots V_{\zeta_1}\vert n\ket\|^2.
  \end{equation}
  Let $Q$ be the orthogonal projector onto the closed subspace spanned by $\{|n\ket: n\in \mathcal N\}$. Then, for any state $\rho\in \cJ_1$,  
  \begin{equation}\label{eq:proba Minfini non pure}
  \P^\rho\left(M_\infty=\frac{Q\rho_{\rm inv}Q}{\tr(\rho_{\rm inv} Q)}\right)=\tr(\rho Q).
  \end{equation}
  Moreover, if $\operatorname{Card}\mathcal N\geq 2$, there exists $\rho\in \mathcal J_1$ such that
  \begin{equation}\label{eq:far from pure states}
  \inf_{t\in \nn^*}\inf_{\phi\in \cH, \|\phi\|=1}\|\rho_t-|\phi\ket\bra\phi|\|_1=1,\quad \pp^\rho\mbox{-a.s.}
  \end{equation}
\end{proposition}

\begin{remark}
  Purification holds if and only if for any maximal $\mathcal N$, we have $\operatorname{Card}\mathcal N=1$.
\end{remark}

\begin{proof} Fix $\rho\in\cJ_1$.
Equation~\eqref{eq:nopur} and the fact that $M_t$ is a function of $N$ imply that for any $t\in \nn^*$, $QM_tQ=\frac{\bra n, M_t n\ket}{\bra  n, \rho_{\rm inv} n\ket} Q\rho_{\rm inv}Q$ where $n$ is an arbitrary element of $\mathcal N$. Hence, using  Corollary~\ref{cor:Mn-convergence} we get
$$QM_\infty Q=\frac{\bra n, M_\infty n\ket}{\bra n, \rho_{\rm inv} n\ket} Q\rho_{\rm inv}Q, \quad \pp^\rho\mbox{-a.s.}$$
Since $M_\infty$ is also a function of the number operator $N$ we have $QM_\infty=M_\infty Q$ so that $M_\infty\propto Q\rho_{\rm inv}Q$ if and only if $\bra n, M_\infty n\ket= 0$ for any $n\notin\mathcal N$. Because $M_\infty$ has trace one, this automatically implies that both statements are also equivalent to $M_\infty=\frac{Q\rho_{\rm inv}Q}{\tr(Q\rho_{\rm inv})}$. Note in particular that $\bra n, M_\infty n\ket = \bra n,\rho_\inv n\ket \tr(\rho_\inv Q) \neq 0$ for any $n\in \mathcal N$.

To prove Equation~\eqref{eq:proba Minfini non pure} we calculate $\P^{|n\ket\bra n|}\left(M_\infty=\frac{Q\rho_{\rm inv}Q}{\tr(\rho_{\rm inv} Q)}\right)$ for all $n$ and then use Lemma~\ref{lem:Prho decomposition}. The key remark is that $\pp^{|n\ket\bra n |}\left(\bra n, M_\infty n\ket\neq 0\right)=1$ for all $n\in\N$. It follows directly from the change of measure formula $\d \pp^{|n\ket\bra n |}=\frac{\bra n, M_\infty n\ket}{\bra n, \rho_{\rm inv} n\ket}\d \pp^{\rho_{\rm inv}}$ which is the generalization to the non-resonant case of the one given in Proposition~\ref{prop:Mn-martingale}, see Equation~\eqref{eq:change measure proof}.

Suppose first $n\notin \mathcal N$. From the above equivalence, $\pp^{|n\ket\bra n |}\left(\bra n, M_\infty n\ket\neq 0\right)=1$ implies 
\begin{equation}\label{eq:proba non pure 1}
 \pp^{|n\ket\bra n|}\left(M_\infty= \frac{Q\rho_{\rm inv}Q}{\tr(\rho_{\rm inv} Q)}\right)=0.
\end{equation}
  
Consider now $n\in \mathcal N$. Following the proof of Proposition~\ref{prop:Mn-martingale}, for any $m\notin\mathcal N$ there exists $s\in \nn^*$ and $\zeta_1,\dotsc,\zeta_s$ such that Equation~\eqref{eq:purificationgeneral} holds (recall $\mathcal N$ is maximal). Hence, for any $m\notin\mathcal N$, $\bra n, M_\infty n\ket\bra m, M_\infty m\ket=0$ $\pp^\rho$-almost surely for any state $\rho$. In particular, using again $\pp^{|n\ket\bra n |}\left(\bra n, M_\infty n\ket\neq 0\right)=1$, we also have $\pp^{|n\ket\bra n |}\left(\bra m, M_\infty m\ket =  0\right)=1$ for all $m\notin\mathcal N$. Therefore, from the above equivalence,  
 \begin{equation}\label{eq:proba non pure 2}
  \pp^{|n\ket\bra n|}\left(M_\infty=\frac{Q\rho_{\rm inv}Q}{\tr(\rho_{\rm inv} Q)}\right)=1.
 \end{equation}

Finally it follows from Equations~\eqref{eq:proba non pure 1}-\eqref{eq:proba non pure 2} and Lemma~\ref{lem:Prho decomposition} that
 $$\pp^{\rho}\left(M_\infty= \frac{Q\rho_{\rm inv}Q}{\tr(\rho_{\rm inv} Q)}\right)=\sum_{n\in \mathcal N}\bra n, \rho n\ket=\tr(\rho Q).$$

\medskip
It remains to prove Equation~\eqref{eq:far from pure states}. Let $m,n\in \mathcal N$ be distinct and chose $\rho=\frac12(|m\ket\bra m|+|n\ket\bra n|)$. Then Equations~\eqref{eq:krausonfock} and Equation~\eqref{eq:nopur} imply 
$$\rho_t(\omega)=\frac12(|m+s_t(\omega)\ket\bra m+s_t(\omega)|+|n+s_t(\omega)\ket\bra n+s_t(\omega)|)$$
for all $t$ and $\pp^\rho$-almost every $\omega$, i.e. $2\rho_t(\omega)$ is a rank $2$ orthogonal projection. A simple computation then shows that, for any such projection $\Pi$ and any unit vector $\phi$, one has $\ds \left\| \frac12\Pi - |\phi\ket\bra\phi| \right\|_1 \geq 1$ with equality if and only if $\phi\in {\rm Range}(\Pi)$. Hence for all $t$ we have 
\[
\inf_{\phi\in \cH, \|\phi\|=1}\|\rho_t-|\phi\ket\bra\phi|\|_1=1,\quad \pp^\rho\mbox{-a.s.}
\]
\end{proof}

\bigskip
We now turn to the analysis of Equation~(\ref{eq:purificationgeneral}). If there is a single resonance then it holds for all $n\neq m$. Indeed, let $n_1$ denote the unique resonance, then
\begin{enumerate}
\item if $n_1\leq m<n$, the result of Lemma~\ref{lem:purificationassumption} holds for those $n$ and $m$, with the same argument choosing all the $\zeta_j$'s equal to $(-,+)$ and $s=m-n_1$.
\item if $n<n_1\leq m$, choosing all the $\zeta_j$'s equal to $(+,-)$ with $s=n_1-n$ we have $\|V_{\zeta_s}\cdots V_{\zeta_1}\vert n\ket\|^2=0$ while $\|V_{\zeta_s}\cdots V_{\zeta_1}\vert m\ket\|^2\neq 0$,
\item if $n<m<n_1$, choosing again all the $\zeta_j$'s equal to $(+,-)$ but with $s=n_1-m$ we have $\|V_{\zeta_s}\cdots V_{\zeta_1}\vert m\ket\|^2=0$ while $\|V_{\zeta_s}\cdots V_{\zeta_1}\vert n\ket\|^2\neq 0$.
\end{enumerate}

\medskip

Suppose now there are infinitely many resonances. If $|n\ket,|m\ket$ belong to the same Rabi sector then the same argument as in case (3) above applies. If they are in two different sectors, let $n_k,n_l\in R(\xi,\eta)$ such that $n_k\leq n<n_{k+1}$ and $n_l\leq m<n_{l+1}$. Choosing first all $\zeta_j$'s equal to $(+,-)$ and $0\leq s\leq \min\{n_{k+1}-n,n_{l+1}-m\}$ we get that Equation~(\ref{eq:purificationgeneral}) holds unless
\[
n_{k+1}-n=n_{l+1}-m \quad \mbox{and} \quad \alpha_{n+j}=\alpha_{m+j}, \ \forall j\in\{0,\ldots, n_{k+1}-n\}, 
\]
and where $\alpha_n$ is the one from Section~\ref{ssec:classicalchain}. On the other hand, choosing all $\zeta_j$'s equal to $(-,+)$ and $0\leq s\leq \min\{n-n_j-1,m-n_k-1\}$ we get that Equation~(\ref{eq:purificationgeneral}) holds unless 
\[
n-n_{k}=m-n_{l} \quad \mbox{and} \quad \alpha_{n-j}=\alpha_{m-j}, \ \forall j\in\{0,\ldots, n-n_{k}\}. 
\]
Combining the two we therefore have the following
\begin{lemma}\label{lem:nopurification} There exists $n\neq m$ such that for any $\zeta_1,\ldots,\zeta_s$ Equation~(\ref{eq:purificationgeneral}) does not hold if and only if there exists $n_k\neq n_l\in \{0\}\cup R(\xi,\eta)$ such that $n_{k+1}-n_k=n_{l+1}-n_l$ and, for all $j\in\{0,\ldots,n_{k+1}-n_k\}$, one has $\alpha_{n_k+j}=\alpha_{n_l+j}$, i.e.
\begin{equation}\label{eq:identicalsectors}
\sin^2(\pi\sqrt{\xi (n_k+j)+\eta})\frac{\xi (n_k+j)}{\xi (n_k+j)+\eta} = \sin^2(\pi\sqrt{\xi (n_l+j)+\eta})\frac{\xi (n_l+j)}{\xi (n_l+j)+\eta}.
\end{equation}
In other words if there exists two Rabi sectors which have the same dimension and the ``same transition probabilities''. In this case~(\ref{eq:purificationgeneral}) fails if and only if $(n,m)\in I_k\times I_l$ is such that $n-n_k=m-n_l$. 
\end{lemma}

Proposition~\ref{prop:nopurification} together with Lemma~\ref{lem:nopurification} give thus necessary and sufficient conditions for a failure of purification. Degenerate systems provide examples where this situation occurs. Indeed, if $n_k,n_l\in N(\xi,\eta)$ then $n_{k+1}-n_k=n_{l+1}-n_l=1$ and for $j\in\{0,1\}$ one has $\alpha_{n_k+j}=\alpha_{n_l+j}=0$. In general, we do not know whether it is possible to find two Rabi sectors of dimension greater than $1$ and such that Equality~(\ref{eq:identicalsectors}) occurs, except in the particular case $\eta=0$ (perfectly tuned cavity).

\medskip
%%%%%%%%%%%%%%%%%%%%%%%%%%%%%%%%%%%%%%%%%%%%%%%%%%%%%%%%%%%%%%%%%%%%
%%%%%%%%%%%%%%%%%%%%%%%%%%%%%%%%%%%%%%%%%%%%%%%%%%%%%%%%%%%%%%%%%%%%

\subsection{Tuned cavity}\label{ssec:tuned}

In this section we suppose that $\eta=0$. This corresponds to the situation where the frequency $\eps_0$ of the cavity is perfectly tuned to the Bohr frequency $\eps$ of the atoms. Using the same argument as the one in \cite{BP09} to study the peripheral spectrum of $\cL$, we prove that the situation of Lemma~\ref{lem:nopurification} can not happen. As a consequence, in a perfectly tuned cavity purification always holds.

As noticed in \cite{BP09} the only possible consecutive resonances are $0$ and $1$. Indeed if $n\geq 1$ is such that $n$ and $n+1$ are resonances there exists $p,q\in \N^*$ such that $\xi n=p^2$ and $\xi(n+1)=q^2$. As a consequence 
\[
\sqrt{\frac{n}{n+1}}=\frac{p}{q}
\]
which contradicts the irrationality of the left-hand side. Hence there does not exist two Rabi sectors of dimension $1$. 

Suppose now there exists $n_k< n_l$ such that $n_{k+1}-n_k=n_{l+1}-n_l=d\geq 2$ and that Equation~(\ref{eq:identicalsectors}) holds for $j=0,\ldots,d$. Taking $j=1$ we have
\[
\sin^2(\pi\sqrt{\xi (n_k+1)}) = \sin^2(\pi\sqrt{\xi (n_l+1)})
\]
hence 
\[
\sqrt{\xi (n_l+1)} \pm\sqrt{\xi (n_k+1)} = p
\]
for some integer $p>0$. Using that $n_l$ is a resonance there exists $q\in\N^*$ such that $\xi n_l=q^2$. Dividing the above identity by $q$ we obtain
\[
\pm \sqrt{\frac{n_k+1}{n_l}} = \frac{p}{q} - \sqrt{\frac{n_l+1}{n_l}}.
\]
Squaring both sides leads to 
\[
\frac{n_k+1}{n_l} = \frac{p^2}{q^2} + \frac{n_l+1}{n_l} -2\frac{p}{q} \sqrt{\frac{n_l+1}{n_l}},
\] 
which contradicts the irrationality of $\ds \sqrt{\frac{n_l+1}{n_l}}$.

\medskip
%%%%%%%%%%%%%%%%%%%%%%%%%%%%%%%%%%%%%%%%%%%%%%%%%%%%%%%%%%%%%%%%%%%%
%%%%%%%%%%%%%%%%%%%%%%%%%%%%%%%%%%%%%%%%%%%%%%%%%%%%%%%%%%%%%%%%%%%%

\subsection{The degenerate case}\label{ssec:degenerate} 

We suppose throughout this section that the system is degenerate, i.e. the set $N(\xi,\eta)$ has at least two elements. Then $N(\xi,\eta)$ is actually a maximal set $\mathcal N$ in the sense of Proposition~\ref{prop:nopurification}. As in this proposition, we denote by $Q$ the orthogonal projection onto the closed linear span of the one dimensional Rabi sectors, i.e.
\begin{equation}\label{def:rabiprojection}
Q := \sum_{n\in N(\xi,\eta)} |n\ket\bra n|.
\end{equation}
The degeneracy assumption amounts to ${\rm Rank}(Q)\geq 2$.

Let $O\in \cO$ be defined by
\begin{equation}\label{def:dontpurifyevent}
O=\left\{\omega\in \Omega: M_\infty(\omega)=\frac{Q\rho_{\rm inv} Q}{\tr(\rho_{\rm inv} Q)}\right\}.
\end{equation}
Proposition~\ref{prop:nopurification} ensure $\P^\rho(O)=\tr(\rho Q)$ for any state $\rho$.

Using exactly the same reasoning as in Section~\ref{ssec:purificationproof}, we omit the details which are left to the reader, we can further deduce that for any state $\rho$ and $\P^\rho$-almost every $\omega\in O$ one has
\[
\lim_{t\to\infty} \left\|\rho_t - U_t \frac{Q \rho Q}{\tr(\rho Q)} U_t^* \right\|_1 = 0,
\]
where we recall that $U_t$ denotes the partial isometry in the polar decomposition of $W_t$. At this point we can actually further use Equation~(\ref{eq:krausoperators}) to make more explicit the action of $U_t$ on $Q\rho Q$. 

For any $\rho$ the quantity $Q\rho Q$ is a linear combination of $|n\ket\bra m|$ with $n,m\in N(\xi,\eta)$. It thus suffices to consider $U_t |n\ket\bra m| U_t^*$ for such $n$ and $m$. Let $n,m\in N(\xi,\eta)$. There exists integers $k,\ell,p,q$ so that 
\[
\sqrt{\xi n+\eta}=k, \ \sqrt{\xi (n+1)+\eta}=\ell, \ \sqrt{\xi m+\eta}=p \ \mbox{ and } \  \sqrt{\xi (m+1)+\eta}=q.
\]
For any $\omega\in O$, and with $N_{t,\omega}(y)$ the number of occurrences of $y\in\cR$ in $\omega\in \Omega$ up to time $t$, Equation~(\ref{eq:krausoperators}) then gives
\[
U_t(\omega)|n\ket\bra m| U_t^*(\omega) = \e^{-i\tau\eps (n-m)t} \left(-1\right)^{(k+p)N_{t,\omega}(--)}\left(-1\right)^{(\ell+q) N_{t,\omega}(++)}|n\ket\bra m|.
\]
From Equation~\eqref{eq:proba non pure 1} we infer that $\pp^{|k\ket\bra k|}(O)=0$ for any $k\not\in N(\xi,\eta)$. On the other, for $n\in N(\xi,\eta)$, it follows from $\alpha_{n}=\alpha_{n+1}=0$ that $\P^{|n\ket\bra n|}(N_{t,\omega}(\pm\mp))=0$. Using Lemma~\ref{lem:Prho decomposition} we derive that, for $\pp^\rho$-almost every $\omega\in O$, $N_{t,\omega}(\pm\mp)=0$ hence $N_{t,\omega}(--)+N_{t,\omega}(++)=t$.

But $p+k$ has the parity $p^2-k^2$ hence of $\xi(m-n)$, and similarly for $\ell+q$, so we finally have
\[
U_t(\omega)|n\ket\bra m| U_t(\omega)^*  = \e^{-i(\tau\eps+\xi\pi) (n-m)t} |n\ket\bra m| = \e^{-i(\eps\tau+\pi\xi)Nt} |n\ket\bra m| \e^{i(\eps\tau+\pi\xi)Nt}.
\]
We can summarize the results of the degenerate case in the following. 
\begin{proposition} Suppose the system is degenerate and let $Q$ and $O$ be defined by~(\ref{def:rabiprojection})-(\ref{def:dontpurifyevent}). For any initial state $\rho$, we have $\P^\rho(O)=\tr(\rho Q)$ and the quantum trajectory starting from $\rho$ satisfies 
\[
\lim_{t\to \infty} \left\|\rho_t - \e^{-i(\eps\tau+\pi\xi)Nt} \frac{Q \rho Q}{\tr(\rho Q)} \e^{i(\eps\tau+\pi\xi)Nt} \right\|_1=0,
\]
for $\P^\rho$-almost every $\omega\in O$.
\end{proposition}

%%%%%%%%%%%%%%%%%%%%%%%%%%%%%%%%%%%%%%%%%%%%%%%%%
%%%%%%%%%%%%%%%%%%%%%%%%%%%%%%%%%%%%%%%%%%%%%%%%%
%%%%%%%%%%%%%%%%%%%%%%%%%%%%%%%%%%%%%%%%%%%%%%%%%
%%%%%%%%%%%%%%%%%%%%%%%%%%%%%%%%%%%%%%%%%%%%%%%%%

\end{document}